\documentclass[11pt]{article}

\usepackage{amsmath}
\usepackage{amsthm}
\usepackage{graphicx}
\usepackage[linesnumbered,ruled,vlined]{algorithm2e}
\usepackage[noend]{algpseudocode}
\usepackage{subfig}
\usepackage{url}
\usepackage{menukeys}
\usepackage{wrapfig}

\usepackage[margin=0.7in]{geometry}

\newcommand{\algoname}[1]{\uppercase{\texttt{#1}}}

\newtheorem{problem}{Problem}
\newtheorem{definition}{Definition}
\newtheorem{theorem}{Theorem}
\newtheorem{lemma}{Lemma}

\newcommand{\runningtime}{The last marker indicates the termination of the algorithm at which the algorithm is not improving any more. Each intermediate point indicates the termination of the algorithm at a particular iteration before convergence.}

\date{}

\begin{document}
\sloppy

\title{ELRUNA : Elimination Rule-based Network Alignment}

\author{
  Zirou Qiu*\\
  \textit{Clemson University} \\
  \texttt{zirouq@clemson.edu}
  \and
  Ruslan Shaydulin\\
  \textit{Clemson University} \\
  \and
  Xiaoyuan Liu\\
  \textit{University of Delaware} \\
  \and
  Yuri Alexeev\\
  \textit{Argonne National Laboratory} \\
  \and
  Christopher S. Henry\\
  \textit{Argonne National Laboratory} \\
  \and
  Ilya Safro\\
  \textit{University of Delaware} \\
  \texttt{isafro@udel.edu}
}
\maketitle
\begin{abstract}
{Networks model a variety of complex phenomena across different domains. In many applications, one of the most essential tasks is to align two or more networks to infer the similarities between cross-network vertices and to discover potential node-level correspondence. In this paper, we propose \algoname{elruna} (\textbf{el}imination \textbf{ru}le-based \textbf{n}etwork \textbf{a}lignment), a novel network alignment algorithm that relies exclusively on the underlying graph structure. Under the guidance of the elimination rules that we defined, \algoname{elruna} computes the similarity between a pair of cross-network vertices iteratively by accumulating the similarities between their selected neighbors. The resulting cross-network similarity matrix is then used to infer a permutation matrix that encodes the final alignment of cross-network vertices. In addition to the novel alignment algorithm, we  improve the performance of \textit{local search}, a commonly used postprocessing step for solving the network alignment problem, by introducing a novel selection method \algoname{rawsem} (\textbf{ra}ndom-\textbf{w}alk-based \textbf{se}lection \textbf{m}ethod) based on the propagation of vertices' mismatching across the networks. The key idea is to pass on the initial levels of mismatching of vertices throughout the entire network in a random-walk fashion.
Through extensive numerical experiments on real networks, we demonstrate that \algoname{elruna} significantly outperforms the state-of-the-art alignment methods in terms of alignment accuracy under lower or comparable running time. Moreover, \algoname{elruna} is robust to network perturbations such that it can maintain a close to optimal objective value under a high level of noise added to the original networks. Finally, the proposed \algoname{rawsem} can further improve the alignment quality with a smaller number of iterations compared with the naive local search method.\\
{\bf Reproducibility}: The source code and data are available at \url{https://tinyurl.com/uwn35an}.}
\end{abstract}

% -------------------------- 
%       Introduction       -
% -------------------------- 
\section{Background and Motivation}
Networks encode rich information about the relationships among entities, including friendships, enmities, research collaborations, and biological interactions~\cite{koutra2017individual}. 
The network alignment problem occurs across various domains. Given two networks, many fundamental data-mining tasks involve quantifying their structural similarities and discovering potential correspondences between cross-network vertices~\cite{final}. For example, by aligning protein-protein interaction networks,  we can discover functionally conserved components and identify proteins that play similar roles in networked biosystems~\cite{netal}. In the context of marketing,  companies often link similar users across different networks in order to recommend products to potential customers~\cite{final}. Furthermore, the network alignment problem  exists in fields such as computer vision~\cite{netalign}, chemistry~\cite{natali}, social network mining~\cite{final}, and economy~\cite{multilevel}.

\par In general, network alignment aims to map\footnote{We use terms \textit{map} and \textit{align} interchangeably throughout the paper.} vertices in one network to another such that some cost function is optimized and pairs of mapped cross-network vertices are similar~\cite{final}. While the exact definitions of similarities are problem dependent, they often reveal some resemblance between structures of two networks and/or additional domain information such as similarities between DNA sequences~\cite{netal}. Formally, we define the network alignment problem as follows.

\begin{problem} [\textbf{Network Alignment Problem}]
Given two networks with underlying undirected, unweighted graphs $G_1 = (V_1, E_1)$ and $G_2 = (V_2, E_2)$ with $|V_1| = |V_2|$ (this constraint is trivially satisfied by adding dummy 0-degree nodes to the smaller network).\footnote{Note that the requirement $|V_1| = |V_2|$ is introduced only to make $\bf{P}$ a square matrix. For the simple computation of the objective,  0-degree dummy nodes do not contribute to it. In later discussion and for the implementation of the algorithm, we do not require $|V_1| = |V_2|$.}
Let $\textbf{A}$ and $\textbf{B}$ be the adjacency matrices of $G_1$ and $G_2$, respectively. The goal is to find a permutation matrix $\textbf{P}$ that minimizes the cost function:
\begin{equation} \label{eq: objective}
    \min_{\textbf{P}} ~~~-trace(\textbf{P}^T\textbf{A}\textbf{P}\textbf{B}^T),
\end{equation}
where $\textbf{P}$ encodes the bijective mappings between $V_1$ and $V_2$ for which $\textbf{P}_{i,u} = 1$ if $i \in V_1$ is aligned with $u \in V_2$, and $\textbf{P}_{i,u} = 0$ otherwise.

\par An equivalent problem is to maximize the number of \textbf{\textit{conserved}} edges in $G_1$, for which an edge $(i, j) \in E_1$ is conserved if $\textbf{P}_{i, u} = 1, \; \textbf{P}_{j, v} = 1$ and $(u, v) \in E_2$. 
\end{problem}

\par The problem above  is a special case of the \textit{quadratic assignment problem} (QAP), which is known to be NP-hard~\cite{approx_qap}. The network alignment problem can also be considered as an instance of a \textit{subgraph isomorphism problem}~\cite{liu2017novel}. Because of its hardness, many heuristics have been developed to solve the problem by relaxing the integrality constraints. Typically, the existing methods first compute the similarity between every pair of cross-network vertices by iteratively accumulating similarities between pairs of cross-network neighbors and then inferring the alignment of cross-network nodes by solving variants of the maximum weight matching problem~\cite{survey}.

\par Many  approaches provide insights into the potential correspondence between cross-network vertices; however, they still exhibit several limitations. First, under the setting of some previous methods~\cite{final, multilevel, bigalign, netalign, isorank, liao2009isorankn, modulealign, hubalign, feizi2019spectral, memivsevic2012c}, computing similarity between $i \in V_1$ and $u \in V_2$ is a process of accumulating the similarities between \textbf{all} pairs of their cross-network neighbors. In other words, while computing the similarity between $i$ and $u$, each of their neighbors contributes \textbf{multiple times}. This might lead to an unwanted case where $i$ has a high similarity score with $u$ simply because $u$ is a high-degree node. Thus they have many pairs of cross-network neighbors that can contribute similarities to $(i, u)$. In addition, this setting  makes it difficult to effectively penalize the degree difference between $i$ and $u$; and after normalization, the resulting similarity is diluted because of the inclusion of many ``noisy'' similarities. 

%\begin{figure}[!h]
\begin{wrapfigure}{r}{0.4\textwidth}
\centering
\includegraphics[width=\linewidth, trim=0cm 2cm 0.5cm 2.5cm, clip]{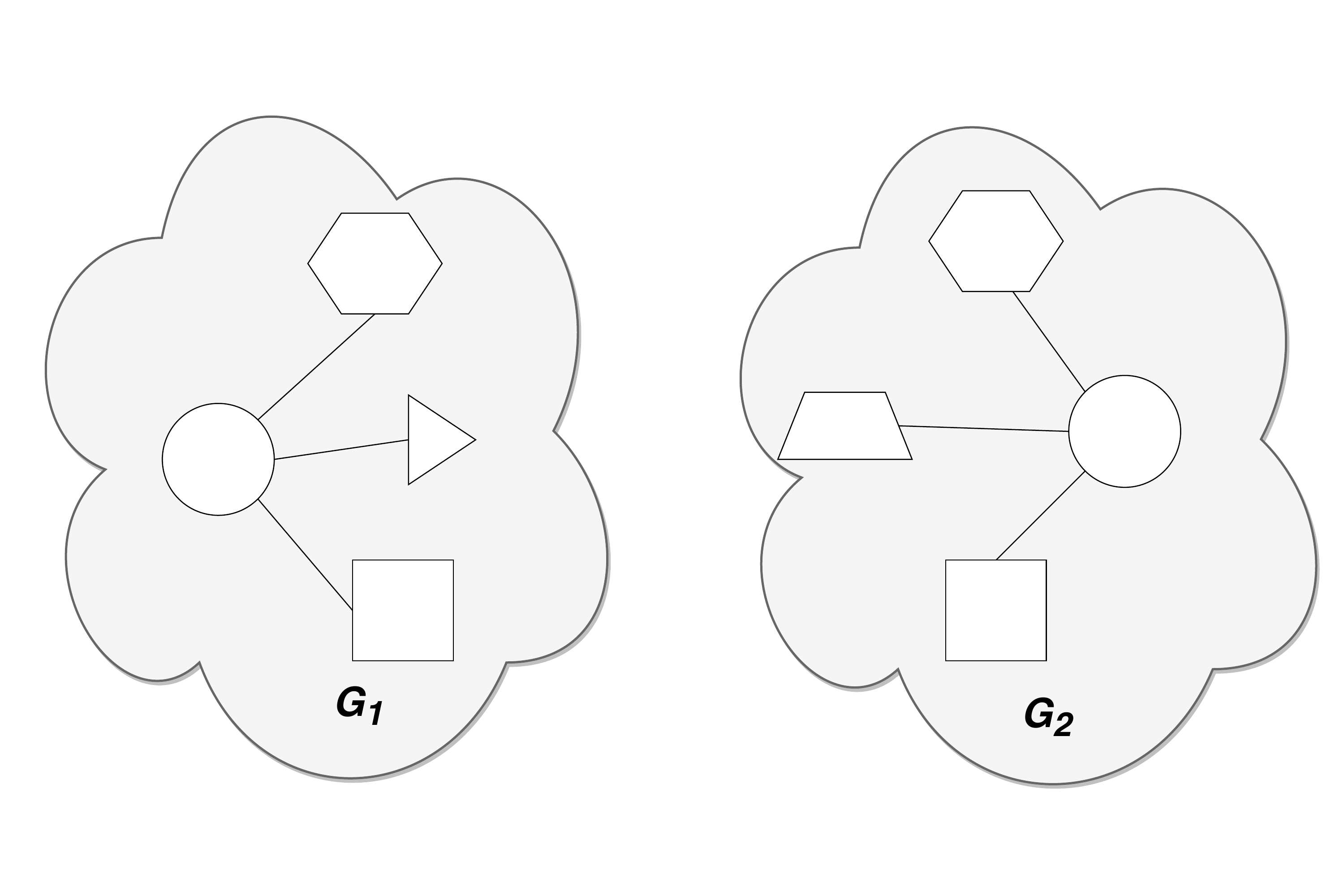}
\caption{Example of the dilution of results}
\label{fig:limitation}
\end{wrapfigure}
%\end{figure}

To illustrate the dilution by noisy similarities, consider two \textbf{unattributed} graphs shown in Figure \ref{fig:limitation}. 
Let the shape of a node denote its ground-truth identity. For example, the circle in $G_1$ should have a higher similarity with the other circle in $G_2$ than nodes with other shapes. When the similarity between two circle nodes gets updated, some methods accumulate the similarities not only between $(hexagon, hexagon)$, $(square, square)$ and $(triangle, trapezoid)$ but also  between $(hexagon, square)$, $(hexagon, trapezoid)$, $(square, hexagon)$, and so on. However, one could argue that the inclusion of the similarity between dissimilar vertices will dilute the result. 
\par Another limitation is that many existing network alignment algorithms rely on non-network information (prior-similarity matrices) to generate high-quality alignments~\cite{netalign, modulealign, feizi2019spectral, natali, final, yasar2018iterative, ghost, du2017first}. However, these methods are useless when no such information is available.

\par \textbf{Our contribution}: We address the network alignment problem by focusing on overcoming the above limitations. The main contributions of this paper are as follows:
\begin{enumerate}
    \item \textbf{\textbf{\algoname{elruna}}: Network Alignment Algorithm}. We propose a novel network alignment algorithm \algoname{elruna} that identifies \textit{globally most similar} pairs of vertices based on the growing \textit{contribution threshold} (both defined later). Such a threshold is used to determine the pairs of cross-network vertices that can contribute similarities while eliminating others. \textbf{To the best of our knowledge, \algoname{elruna} is the first network alignment algorithm that introduces the elimination rule into the process of accumulating similarities.} Another novelty is that \algoname{elruna} solves the network alignment problem by iteratively solving smaller subproblems at the neighborhood scale. Extensive experimental results show that the proposed \algoname{elruna} significantly outperforms the state-of-the-art counterparts under lower or comparable running time. Moreover, \algoname{elruna} can maintain a close-to-optimal objective value under a high level of noise added to the original networks. What makes \algoname{elruna} even more competitive is that it discovers high-quality alignment results relying solely  on the topology of the networks (without the help of non-network information).
    \item \textbf{\algoname{rawsem}: Selection Method for Local Search.} We introduce a novel selection method \algoname{rawsem} for a local-search procedure that narrows the search space by locating mismatched vertices. The proposed method first quantifies the initial amount of mismatching of each vertex and then propagates the values of mismatching throughout the network in a \texttt{PageRank}~\cite{pagerank} fashion. After convergence, vertices with a high level of mismatching are selected into the search space. To the best of our knowledge, \algoname{rawsem} is the first random-walk-based \textbf{selection method} for the local search scheme in solving the network alignment problem.
    
    \item \textbf{Evaluations}. We conduct extensive experiments to analyze the effectiveness and efficiency of \algoname{elruna}. We compare \algoname{elruna} with eight carefully chosen  state-of-the-art baselines that are superior to many other algorithms in both quality and running time. Our experiments cover three real-world scenarios: (1) \textit{self-alignment with and without additional noise}, (2) \textit{alignment between homogeneous networks}, and (3) \textit{alignment between heterogeneous networks}. The results show that the proposed \algoname{elruna} (even before applying postprocessing local search) already significantly outperforms all the baseline methods. At the same time, the proposed \algoname{rawsem} can further improve the objective values to an optimum with a dramatically decreased number of iterations compared with the naive local search.
\end{enumerate}

% --------------------------------
% -        Related Work          -
% --------------------------------
\section{Related Work}
Extensive research has been conducted to solve the network alignment problem. In most methods, the underlying intuition is that \emph{two cross-network vertices are similar if their cross-network neighbors are similar}. The limitations of the existing works are discussed in the introduction section.

\par \textbf{Unattributed network alignment:} Koutra et al.~\cite{bigalign} formulated the bipartite network alignment problem as a quadratic assignment problem and proposed \texttt{Big-align}, which is an iterative improvement algorithm to find the optimum.  Bayati et al.~\cite{netalign} introduced \texttt{NetAlign}, which treats the network alignment problem as an integer
quadratic program and solves it using the belief propagation. Xu et al.~\cite{liu2017novel} solved the network alignment problem by projecting the problem into the domain of computational geometry while preserving the topology of the graphs. Then they used the rigid transformations approach to compute the similarity scores. Feizi et al. introduced \texttt{EigenAlign}~\cite{feizi2019spectral}, which formulates the network alignment problem with respect to not only the number of conserved edges but also non-conserved edges and neutral edges. The authors then solved the problem as an eigenvector computation problem that finds the eigenvector with the largest eigenvalue. 

\par \textbf{The MAGANA family}: Saraph and Milenkovic introduced \texttt{MAGANA}~\cite{saraph2014magna}, which uses a genetic algorithm-based local search to solve the alignment problem. Later, they introduced \texttt{MAGANA++}\cite{vijayan2015magna++}, which extends \texttt{MAGANA} with parallelization.  Milenkovic et al.~\cite{vijayan2017alignment} addressed the dynamic network alignment problem where the structures of the networks evolve over time. They proposed \texttt{DynaMAGNA++}~\cite{vijayan2017alignment}, which conserves dynamic edges and nodes.  %\texttt{SANA} is another local search-based network alignment algorithm. In contrast to the \texttt{MAGANA} family, \texttt{SANA} uses the simulated annealing as the underlying metaheuristic algorithm which is very slow. \added[id=is]{I don't think SANA deserves citation after their letter} \added[id=zq]{my concern was that he might be a reviwer, despite of our list of "unwanted reviwers".}

\textbf{Attributed network alignment}: Klau~\cite{natali} \ formulated the problem using the maximum weight trace and suggested a Lagrangian relaxation approach. Zhang and Tong~\cite{final} tackled the attributed network alignment problem for which vertices have different labels. They dropped the topological consistency assumption and solved the problem by using attributes as alignment guidance. Du et al.~\cite{du2017first} also addressed the attributed network alignment problem where the underlying graphs are evolving. They formulated the problem as a Sylvester equation and solved it in an incremental fashion with respect to the updates of networks. In another work  Zhang et al.~\cite{multilevel} solved the multilevel network alignment problem based on a coarsening and uncoarsening scheme. By coarsening the network into multiple levels, they discovered not only the node correspondence of the original network but also the cluster-level correspondence with different granularities. Such coarsening-uncoarsening methods have been  successful in solving various cut-based optimization problems on graphs~\cite{ron2011relaxation,safro2006graph,safro2015advanced} but, to the best of our knowledge, are used for the first time for network alignment.

Heimann et al. introduced \texttt{REGAL}~\cite{heimann2018regal}, which tackles the network alignment problem from a node representation learning perspective. By leveraging the low-rank matrix approximation method, they extracted the node embeddings, then constructed the alignment of vertices based on the similarity between embeddings of cross-network vertices. In addition to finding the one-to-one mapping of vertices, \texttt{REGAL} can also identify the top-$\alpha$ potential mappings for each vertex.  Heimann et al. proposed \texttt{HashAlign}~\cite{heimann2018hashalign}, which solves the multiple network alignment problem also based on node representation learning for which each node-feature vector encodes topological features and attributed features.

\textbf{Biological network alignment}:
\texttt{IsoRank}~\cite{isorank} is an alignment algorithm that is equivalent to PageRank on the Kronecker product of two networks. \texttt{IsoRankN}~\cite{liao2009isorankn} extends \texttt{IsoRank} by applying spectral graph partitioning to align multiple networks simultaneously. \texttt{Hubalign}~\cite{hubalign}  involves computing topological importance for each node. Two cross-network vertices have similar scores if they play similar roles in the networks, such as hubs or nodes with high betweenness centralities. \texttt{NETAL}~\cite{netal}  introduces the concept of interaction scores between each pair of cross-network vertices that are estimations of the number of conserved edges. \texttt{ModuleAlign}~\cite{modulealign} combines the topological information with non-network information such as protein sequence for each vertex and produces an alignment that resembles both topological similarities and sequence similarities.  \texttt{GHOST}~\cite{ghost} is another biological network alignment algorithm based on the graph spectrum. \texttt{GHOST} determines the similarities between cross-network vertices based on the similarities between the topological signatures of each vertex; each signature is obtained by computing the spectrum of the k-egocentric subgraph of each vertex.   \texttt{C-GRAAL}~\cite{memivsevic2012c} is a member of the \texttt{GRAAL} family that iteratively computes the cross-network similarities based on the \textit{combined neighborhood density} of each node.

%\added[id=is]{What is missing after all these algorithms is a paragraph with summary of what are the weak points of each of these methods or general weak points for groups of them; alternatively, at the end of each paragraph you can say how do we address their weakness in a sentence or two} \added[id=zq]{this is address in the introduction section.}

% -----------------------------------
% -       Proposed Algorithm        -
% -----------------------------------
\section{\algoname{elruna} : Elimination Rule-Based Network Alignment}
In this section, we introduce the network alignment algorithm \algoname{elruna}. In Table \ref{tab:symbol} we first provide notation used throughout this paper. Next we define three rules that serve as a guide for our algorithm.  We then introduce the pseudocode of \algoname{elruna} and analyze its running time.

\begin{center}
\captionof{table}{Notation} \label{tab:symbol} 
 \begin{tabular}{||l l||} 
 \hline
 \textbf{Symbol} & \textbf{Definition}\\ [0.5ex] 
 \hline \hline
 $G_1 = (V_1, E_1), G_2 = (V_2, E_2)$ & the two networks \\
 $\mathbf{A}, \mathbf{B}$ &  adjacency matrices of $G_1$ and $G_2$ \\
 \hline
 $n_1$, $n_2$ & number of nodes in $G_1$ and $G_2$ \\
 $m_1$, $m_2$ & number of edges in $G_1$ and $G_2$ \\
 \hline
 $i_{(1)}$, $j_{(1)}$ & example nodes in $G_1$ \\
 $u_{(2)}$, $v_{(2)}$ & example nodes in $G_2$ \\
 \hline
 $N(i)$ &  set of neighbors of node $i$ (does not include $i$)\\
 \hline
 $\mathbf{S}$ &  $n_1 \times n_2$ cross-network similarity matrix \\
 $\mathbf{P}$ &  $n_1 \times n_2$ alignment matrix \\
 \hline
 $f : V_1 \rightarrow V_2$ &  injective alignment function. \\
 \hline
 $t_{max}$ & maximum number of iterations \\
 \hline
 
\end{tabular}

\end{center}

\par Given two undirected networks with underlying graphs $G_1 = (V_1, E_1)$ and $G_2 = (V_2, E_2)$, let $|V_1| = n_1$, $|V_2| = n_2$, $|E_1| = m_1$ and $|E_2| = m_2$. Without loss of generality, assume $n_1 \leq n_2$. Nodes in each network are labeled with consecutive integers starting from $1$. 
\par Throughout the paper, we use bold uppercase letters to represent matrices and bold lowercase letters to represent vectors. We use a subscript over a node to indicate the network it belongs to, for example, $i_{(1)} \in V_1$. We use a superscript over vectors and matrices to denote the number of iterations. Let $N(i_{(1)})$ denote the set of neighbors of vertex $i_{(1)}$. Let $\mathbf{S}$ be the $n_1 \times n_2$ cross-network similarity matrix, where $\mathbf{S}_{i,u}$ encodes the similarity score between $i_{(1)}$ and $u_{(2)}$. Note that $i_{(1)}$ and $u_{(2)}$ do not carry superscripts for matrix or vector indexing. Let $f : V_1 \rightarrow V_2$ denote the injective alignment function for which $f(i_{(1)}) = u_{(2)}$ if $\textbf{P}_{i,u} = 1$ ($\textbf{P}$ is defined in Equation (\ref{eq: objective})). Function $f$ is injective because $n_1$ could be less than $n_2$ ($f$ is bijective when $n_1 = n_2$). Let $t_{max}$ denote the maximum number of iterations of our algorithm. We always set $t_{max}$ equal to the larger diameter of the two networks. 

\par The proposed \algoname{elruna} relaxes the combinatorial constraints of $\mathbf{P}$ defined in Problem \ref{eq: objective}; that is, it iteratively computes a similarity matrix $\mathbf{S}$ instead of directly finding a permutation matrix $\mathbf{P}$. In general, our proposed algorithm \algoname{elruna} is a two-step procedure: 
\begin{enumerate}
    \item[](1) \textbf{Similarity computation:} Based on the elimination rules, \algoname{elruna} iteratively updates the \textit{cross-network similarity matrix} $\mathbf{S}$, which encodes similarities between cross-network vertices.
    \item[] (2) \textbf{Alignment:} Based on the converged $\mathbf{S}$, \algoname{elruna} computes the 0-1 alignment matrix $\mathbf{P}$, which encodes the final alignment of cross-network vertices. 
\end{enumerate}

Note that the main difference between most alignment algorithms is how they compute the similarities between vertices. The alignment task does not require a complex alignment method (the second step) if its similarity computation step produces high-quality alignment matrices. Therefore, the main focus of \algoname{elruna} is to compute high-quality alignment matrices.

\subsection{Step 1: Similarity Computation}
We introduce three rules that serve as the guidance of the proposed algorithm. We then present the pseudocode for computing the similarity matrix $\mathbf{S}$. Overall, the proposed \algoname{elruna} computes the similarities between cross-network vertices by updating ${\bf S}$ iteratively. 
\par We first provide an intuitive concept of what it means for a vertex to \textbf{\textit{contribute}} to the similarity computation process. Given a pair of cross-network vertices ($i_{(1)}$, $u_{(2)}$), let $j_{(1)}$ be a neighbor of $i_{(1)}$. At the $k$th iteration of the algorithm, $j_{(1)}$ is said to \textit{\textbf{contribute}} to the computation of $\mathbf{S}^{(k)}_{i,u}$ if we accumulate (defined later) the similarity between $j_{(1)}$ and a neighbor $v_{(2)}$ of $u_{(2)}$ into $\mathbf{S}^{(k)}_{i,u}$. By the same token, the pair ($j_{(1)}$, $v_{(2)}$) is also said to \textbf{\textit{contribute}} to the computation of $\mathbf{S}^{(k)}_{i,u}$. We now provide three essential definitions that are the backbone of \algoname{elruna}.

\begin{definition} [\textbf{Conserved Vertices and Edges}]
Given a vertex $i_{(1)}$ and its aligned vertex $u_{(2)} = f(i_{(1)})$, let $j_{(1)} \in N(i_{(1)})$ be a neighbor of $i_{(1)}$ and $v_{(2)} = f(j_{(1)})$ be the aligned vertex of $j_{(1)}$. Node $j_{(1)}$ is a \textbf{\textit{conserved}} neighbor of $i_{(1)}$ if $v_{(2)} \in N(u_{(2)})$. Under this scenario, $v_{(2)}$ is also a conserved neighbor of $u_{(2)}$. An edge is \textbf{\textit{conserved}} if its incident vertices are conserved neighbors of each other.
\end{definition}

\begin{definition} [\textbf{Best Matching}]
A vertex $u_{(2)}$ is the \textbf{\textit{best matching}} of a vertex $i_{(1)}$ if aligning $i_{(1)}$ to $u_{(2)}$ maximizes the number of conserved neighbors of $i_{(1)}$ in comparison with aligning $i_{(1)}$ to other nodes in $V_2$. The best matching of $u_{(2)}$ is defined in the same fashion.
\end{definition}

\begin{definition} [\textbf{Globally Most Similar}]
At the $k$th iteration, a vertex $u_{(2)}$ is \textbf{\textit{globally most similar}} to a vertex $i_{(1)}$ if $\text{\bf{S}}^{(k)}_{iu} \geq \text{\bf{S}}^{(k)}_{iv}, \; \forall v_{(2)} \in V_2$. Likewise, $i_{(1)}$ is globally most similar to $u_{(2)}$ if $\text{\bf{S}}^{(k)}_{iu} \geq \text{\bf{S}}^{(k)}_{ju}, \; \forall j_{(1)} \in V_1$.
\end{definition}

Note that $i_{(1)}$ being globally most similar to $u_{(2)}$  does not imply that $u_{(2)}$ is also globally most similar to $i_{(1)}$. Also, it is possible that $u_{(2)}$ is globally most similar to more vertices in $G_1$, but they each have a different similarity value with $u_{(2)}$. Given two vertices $i_{(1)}$ and $j_{(1)}$ in $G_1$, we can have $\text{\bf{S}}^{(k)}_{iu} \neq \text{\bf{S}}^{(k)}_{ju}$, but among all vertices in $G_2$, $u$ has the highest similarity with both $i_{(1)}$ and $j_{(1)}$, respectively. 

\par Theorem \ref{them:main}  shows that the objective defined in Equation (\ref{eq: objective}) is minimized when all nodes are aligned to their best matchings (if possible).

\begin{theorem}\label{them:main}
Given an alignment matrix $\text{\bf{P}}$ for which all nodes are aligned to their best matchings, $\text{\bf{P}}$ is an optimal solution of Equation \ref{eq: objective}.
\end{theorem}

\begin{proof}
Suppose for the sake of contradiction that there exists a better alignment matrix $\Bar{\textbf{P}} \neq \textbf{P}$ such that 
\begin{equation} \label{eq:best_matching_optimal}
    trace(\Bar{\textbf{P}}^T\textbf{A}\Bar{\textbf{P}}\textbf{B}^T) > trace(\textbf{P}^T\textbf{A}\textbf{P}\textbf{B}^T).
\end{equation}

Let $\textbf{B}'$ and $\textbf{B}''$ denote $\Bar{\textbf{P}}^T\textbf{A}\Bar{\textbf{P}}$ and $\textbf{P}^T\textbf{A}\textbf{P}$, respectively. Inequality (\ref{eq:best_matching_optimal}) implies that 
\begin{equation} \label{eq:r_c_ineq}
    \textbf{B}'_{i, *} \textbf{B}^{T}_{*,i} > \textbf{B}''_{i, *} \textbf{B}^{T}_{*,i}, \; \exists \; i \in {V}_1 ,
\end{equation}
where $\textbf{B}^{T}_{i, *}$ and $\textbf{B}^{T}_{*,i}$ denote the $i$th row and column of $\textbf{B}^{T}$, respectively. However, the inequality (\ref{eq:r_c_ineq}) implies that in $\mathbf{P}$ there exists a vertex $i_{(1)}$ that is not aligned with its best matching, which is a contradiction.
\end{proof}

\par By Theorem \ref{them:main}, in order to minimize the objective value, one wants nodes to be aligned with their best matchings. \textbf{Algorithm \algoname{elruna} heuristically attempts to ensure that the node that is globally most similar to $i_{(1)}$, provided by the similarity matrix $\mathbf{S}$, corresponds to the best matching of $i_{(1)}$.} Then the alignment process is simply to map each node $i_{(1)}$ in $G_1$ to its globally most similar vertex $u_{(2)}$ in $G_2$. 

\subsubsection{Rule 1 -- level-one elimination}
\par While computing similarity between a pair of cross-network vertices, under the setting of \algoname{elruna}, a neighbor cannot contribute twice. Given a pair of vertices $(i_{(1)}, u_{(2)})$, we consider computing their similarity as a process of aligning their neighbors. In other words, \textbf{a pair of cross-network neighbors $(j_{(1)}, v_{(2)})$ can contribute its similarity to $\mathbf{S}_{i,u}$ only if $j_{(1)}$ is qualified to be aligned with $v_{(2)}$}. Ideally, a pair of cross-network neighbors can be aligned (and therefore qualified to contribute) if at least one of them is globally most similar to the other. As a result, $i_{(1)}$ and $u_{(2)}$ have a higher similarity if they have more neighbors that can be aligned, which also minimizes the objective defined in Equation (\ref{eq: objective}). The injective nature of alignments leads to our first rule.
 
 \begin{enumerate}
    \item[] \textbf{Rule 1.} At the $k$th iteration of the algorithm, given a pair of cross-network vertices $(i_{(1)}, u_{(2)})$, a neighbor $j_{(1)}$ of $i_{(1)}$ can contribute its similarity (with a neighbor of $u_{(2)}$) to $\mathbf{S}^{(k)}_{i,u}$ \textbf{at most once}. Similarly, a neighbor $v_{(2)}$ of $u_{(2)}$ can contribute its similarity (with a neighbor of $i_{(1)}$) to $\mathbf{S}^{(k)}_{i,u}$ \textbf{at most once}. 
\end{enumerate}

%\begin{figure}[!h]
\begin{wrapfigure}{r}{0.4\textwidth}
\centering
\includegraphics[width=\linewidth]{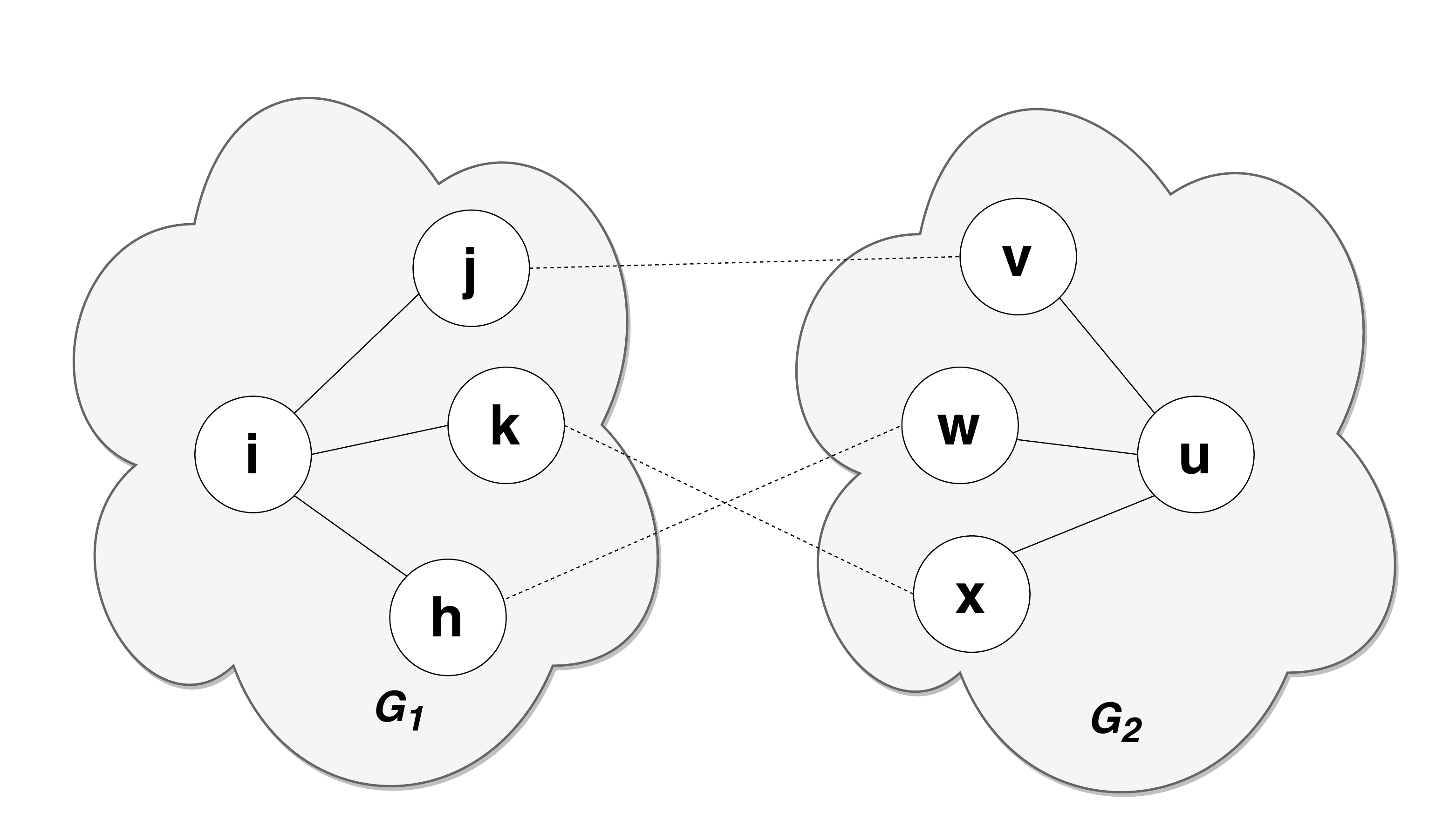}
\caption{Example of Rule 1} 
\label{fig:rule1}
\end{wrapfigure}
%\end{figure}

Figure (\ref{fig:rule1}) illustrates an example graph under Rule 1 where dashed lines indicate the pairs of cross-network neighbors that contribute to the computation of $\mathbf{S}^{(k)}_{i,u}$. For example, if the similarity between $j_{(1)}$ and $v_{(2)}$ is accumulated into $\mathbf{S}^{(k)}_{i,u}$, then the similarity between ($j_{(1)}$, $w_{(2)}$), ($j_{(1)}$, $x_{(2)}$), ($k_{(1)}$, $v_{(2)}$), and ($h_{(1)}$, $v_{(2)}$) can no longer contribute to $\mathbf{S}^{(k)}_{i,u}$. 
\par Formally, Rule 1 decreases the number of pairs of cross-network contributing neighbors from $|N(i_{(1)})| |N(u_{(2)})|$ to at most $\min \{|N(i_{(1)})|, |N(u_{(2)})|\}$. This setting provides an effective way to penalize the degree differences (as shown later). Additionally, it reduces the amount of ``noisy'' similarities included during the iteration process.

\subsubsection{Rule 2 -- level-two elimination} Prior to defining Rule 2, we assume that Rule 1 is satisfied. As we have defined previously, a pair of cross-network neighbors can be aligned if at least one of them is globally most similar to the other. Because of the iterative nature of \algoname{elruna}, however, during the first several iterations of the algorithm, the computed similarities are less revealing of the true similarities between vertices. In other words, we are less certain about whether the node that is globally most similar to $j_{(1)}$ is indeed its best match. Therefore, for each iteration,  allowing only most similar pairs of neighbors globally to contribute while discarding others might fail to accumulate valuable information.

\par We want to relax this constraint. We observe that as we proceed with more iterations, the reliability of similarities increases. To model this increase of confidence, we define \textit{growing} thresholds such that under Rule 1, \textbf{pairs of neighbors whose similarities are greater than their thresholds are allowed to contribute}. Such thresholds are low at the first iteration and grow gradually as we proceed with more iterations. 
\par We define vectors $\mathbf{b1}$ and $\textbf{b2}$ for two networks respectively such that 
\[
\textbf{b1}^{(k)}_i = \max_{u \in V_2} \textbf{S}^{(k)}_{i,u} \text { and } \textbf{b2}^{(k)}_{u} = \max_{i \in V_1} \textbf{S}^{(k)}_{i, u}.
\]
Informally, $\textbf{b1}^{(k)}_i$ is the similarity between $i_{(1)}$ and its globally most similar vertex at the $k$th iteration. Additionally, we define \textit{\textbf{contribution-threshold vectors}} $\textbf{c1}$ and $\textbf{c2}$ for two networks. Given a pair of vertices $(i_{(1)}, u_{(2)})$, a pair of their cross-network neighbors $(j_{(1)}, v_{(2)})$ can  contribute similarity to $\mathbf{S}^{(k+1)}_{i,u}$ only if $\textbf{S}^{(k)}_{j, v} \geq \min\{ \textbf{c1}^{(k)}_j \; ,\textbf{c2}^{(k)}_v\}$.  At the same time, such thresholds grow as the algorithm proceeds with more iterations.

 \par Computing the similarity between $(i_{(1)}, u_{(2)})$ iteratively can be seen as a process of gathering information (regarding the similarity) from other nodes in a breadth-first search manner such that $\textbf{S}^{(k)}_{i, u}$ is computed based on similarities between cross-network nodes that are within distance $k$ away from $i_{(1)}$ and $u_{(u)}$. A node $j_{(1)}$ is said to be \textit{visited} by $i_{(1)}$ at the the $k$th iteration if $j_{(1)}$ is within distance $k$ away from $i_{(1)}$. \textbf{We use the fraction of visited nodes after each iteration as a measure to model the increase of the contribution threshold}. We now define the contribution threshold formally.

\begin{definition} [\textbf{Contribution Threshold}]
Given $G_1$ with size $n_1$, let \text{\bf{T1}} be the $n_1 \times t_{max}$ matrix for which $\text{\bf{T1}}_{i, k}$ is the fraction of nodes that $i_{(1)}$ has visited after the $k$th iteration. Then the \textbf{contribution threshold} of $i_{(1)}$, denoted by $\text{\bf{c1}}_{i}$, after the $k$th iteration is defined as

\begin{equation} \label{eq:ct}
    \text{\bf{c1}}^{(k)}_i = \text{\bf{b1}}^{(k)}_i \cdot \text{\bf{T1}}_{i,k}.
\end{equation}
\end{definition}

$\mathbf{T1}_{i, 0} = 1 / n_1 $ for all $i_{(1)} \in {V}_1$.  As $k$ approaches $t_{max}$, $\mathbf{T1}_{i, k}$ approaches $1$ and $\textbf{c}^{(k)}_i $ approaches $\textbf{b}^{(k)}_i$. $\mathbf{c2}$ and $\mathbf{T2}$ are defined for $G_2$ in the same fashion.  The pseudocode for computing $\textbf{T1}$ and $\mathbf{T2}$ is shown in Algorithm (\ref{algo:node_coverage}).

\begin{algorithm} [htb!]
\DontPrintSemicolon
% ------------ Input -------------%
\KwIn{${G} = ({V}, {E})$, $t_{max}$}

% ------------ Output -------------%
\KwOut{Contribution threshold matrix $\textbf{T}$ }

% ----------- Initialization of S and b ------------ %
$\textbf{T} \gets n \times t_{max}$ zero-matrix \Comment{$|{V}| = n$} \;

% ----------- Main Iteration ---------------%
\For{$i \in {V}$} {
    $\textbf{d} \gets n \times 1$ vector with all entries equal to $0$ \;
    $\textbf{d}_i = 1$ \;
    $num\_of\_visited\_node \gets$ $1$ \;
    $frontier \gets$ empty list. \;
    $frontier$.\textbf{insert}($i$)\;
    \For{$k \gets 1$ \textbf{to} $t_{max}$} {
        $new\_frontier \gets$ empty list. \;
        \For {$j$ \textbf{in} $frontier$}
        {
            \For{$q$ \textbf{in} ${N}(j)$}
            {
                \If{$\textbf{d}_q$ == $0$}
                {
                    $num\_of\_visited\_node += 1$\;
                    $new\_frontier$.\textbf{insert}($q$)\;
                    $\textbf{d}_q = 1$\;
                }
            }
        }
        $\textbf{T}_{i,k} = \frac{num\_of\_visited\_node}{n}$\;
        $frontier = new\_frontier$\;
    }
}
\Return{\textbf{T}}\;
\caption{Contribution Threshold}
\label{algo:node_coverage}
\end{algorithm}

\par  Let $(j_{(1)}, v_{(2)})$ be a pair of cross-network neighbors of $(i_{(1)}, u_{(2)})$. Without loss of generality, suppose $\textbf{S}^{(k)}_{j,v} \geq \textbf{c1}^{(k)}_j$ but $\textbf{S}^{(k)}_{j,v} < \textbf{c2}^{(k)}_v$. This implies that there must exist a better alignment with higher similarity for $v_{(2)}$ at the $k$th iteration. We still want to accumulate the similarity between $j_{(1)}$ and $v_{(2)}$ to $\textbf{S}^{(k)}_{i,u}$. At the same time, we should also consider the loss of similarity by aligning $j_{(1)}$ to $v_{(2)}$ (recall that we model similarity accumulation as a process of aligning neighbors). We introduce a simple measure called \emph{net similarity}:
\begin{equation}
    \textbf{N}^{(k)}_{jv} = 2 \textbf{S}^{(k)}_{j, v} - \left[\frac{\textbf{S}^{(k)}_{j, v} - \textbf{c1}^{(k)}_{j}}{\textbf{b1}^{(k)}_j - \textbf{c1}^{(k)}_{j}} \cdot (\textbf{b2}^{(k)}_v - \textbf{c2}^{(k)}_{v}) + \textbf{c2}^{(k)}_{v}\right]
\end{equation}
and if $\textbf{S}^{(k)}_{j,v} \geq \textbf{c2}^{(k)}_v$ but $\textbf{S}^{(k)}_{j,v} < \textbf{c1}^{(k)}_j$,

\begin{equation}
    \textbf{N}^{(k)}_{jv} = 2 \textbf{S}^{(k)}_{j, v} - \left[\frac{\textbf{S}^{(k)}_{j, v} - \textbf{c2}^{(k)}_{v}}{\textbf{b2}^{(k)}_v - \textbf{c2}^{(k)}_{v}} \cdot (\textbf{b1}^{(k)}_j - \textbf{c1}^{(k)}_{j}) + \textbf{c1}^{(k)}_{j}\right],
\end{equation}
which leads to our second rule.

\begin{enumerate}
    \item[] \textbf{Rule 2.} Given a pair of cross-network vertices $(j_{(1)}, v_{(2)})$, under Rule 1, the amount of similarity they can contribute to a pair of neighbors $(i_{(1)}, u_{(2)})$ is 
    $$
    \begin{cases}
        \mathbf{S}^{(k)}_{j, v} & \text{if $\textbf{S}^{(k)}_{j,v} \geq \max \{ \textbf{c1}^{(k)}_j, \; \textbf{c2}^{(k)}_v$\} }  \\
        
        \mathbf{N}^{(k)}_{j, v} & \text{if $\min \{ \textbf{c1}^{(k)}_j, \; \textbf{c2}^{(k)}_v\} \leq \mathbf{S}^{(k)}_{j, v}  \leq  \max \{ \textbf{c1}^{(k)}_j, \; \textbf{c2}^{(k)}_v\} $ } \\
        
        0 & \text{otherwise}. \\
    \end{cases}
    $$
\end{enumerate}

\begin{wrapfigure}{r}{0.4\textwidth}
\centering
\includegraphics[width=\linewidth]{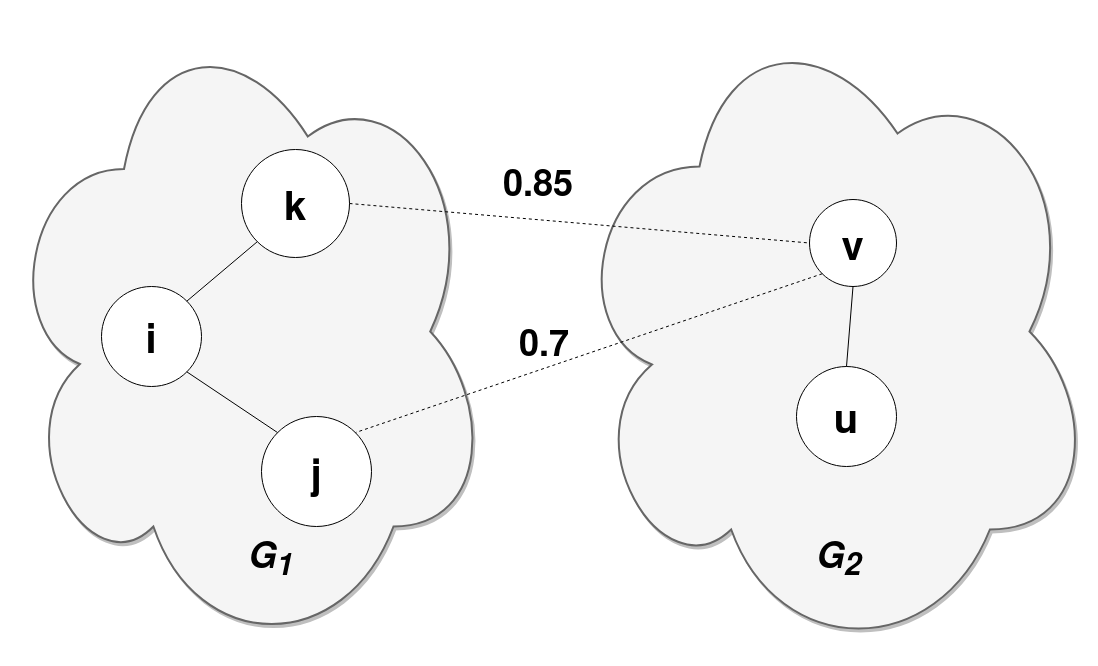}
\caption{Example of Rule 3}
\label{fig:rule3}
\end{wrapfigure}
%\end{figure}

\subsubsection{Rule 3 -- Prioritization} Given the first two rules, it is possible that a neighbor of $u_{(2)}$ is globally most similar to multiple neighbors of $i_{(1)}$ but with different similarities. For example, consider a sample graph shown in Figure~(\ref{fig:rule3}) where $v_{(2)}$ is globally most similar to both $j_{(1)}$ and $k_{(1)}$. Dashed lines indicate the similarities between two vertices.

%\begin{figure}[!h]

\par In this case, we ought to  consider only the pair $(k_{(1)}, v_{(2)})$ and contribute its similarity to $\mathbf{S}_{i,u}$, which leads to the final rule.

\begin{enumerate}
    \item[] \textbf{Rule 3.} During the computation of the similarity between $i_{(1)}$ and $u_{(2)}$, a pair of cross-network neighbors with a \textbf{higher similarity should be given  prior consideration}.
\end{enumerate}

\subsubsection{
The Similarity Computation Algorithm}The general idea of the \algoname{elruna} is to update $\mathbf{S}$, $\mathbf{b1}$ and $\mathbf{b2}$ iteratively based on their values in the previous iteration. \textbf{The initial similarities between each pair of cross-network vertices are uniformly distributed. We set them all equal to $1$.} That is, $\mathbf{S}^{(0)}$ is a matrix of ones. Note that many other algorithms require prior knowledge about the similarities (usually based on non-network information) between vertices, whereas \algoname{elruna} does not require such knowledge.

\par The detailed algorithm is summarized in Algorithm (\ref{algo:ini_solution}). Overall, at the $(k-1)$th iteration, for each pair of cross-network vertices ($i_{(1)}, u_{(2)}$), we first check all pairs of their cross-network neighbors against Rule 2 to determine which pairs are qualified such that the similarity is greater than the contribution threshold of at least one node in the pair. Then we sort those pairs by similarities in descending order, which is needed to follow Rule 3. After sorting, we go over each ($j_{(1)}, v_{(2)}$) in the sorted order, checking $j_{(1)}$ and $v_{(2)}$ against Rule 1. If none of them has contributed to $\mathbf{S}^{(k)}_{i,u}$ before, we accumulate the similarity between $(j_{(1)}, v_{(2)})$ based on Rule 2 and mark $j_{(1)}$ and $v_{(2)}$ as \emph{selected}, which indicates that they can no longer be considered. This step enforces Rule 1. Note that the \textit{selected} neighbors will no longer be selected after we are done computing $\mathbf{S}^{(k)}_{i,u}$.  We now update $\textbf{S}^{(k)}$, $\textbf{b1}^{(k)}$, and $\mathbf{b2}^{(k)}$.

\par After accumulating similarities from neighbors, we normalize it by
$$
    \textbf{S}^{(k+1)}_{iu} = \frac{\text{accumulated similarity}}{\max \{ \sum_{j \in {N}(i)} \textbf{b1}^{(k)}_j, \sum_{v \in {N}(u)} \textbf{b2}^{(k)}_v \}}.
$$

$\textbf{S}^{(k+1)}_{iu}$ is set to $0$ if $\max \{ \sum_{j \in {N}(i)} \textbf{b1}^{(k)}_j, \sum_{v \in {N}(u)} \textbf{b2}^{(k)}_v \} = 0$. This normalization also penalizes the degree of discrepancy between $i_{(1)}$ and $u_{(2)}$ because the maximum number of pairs of cross-network neighbors that can contribute similarity to ($i_{(1)}, u_{(2)}$) is upper bounded by the smaller degree between $i_{(1)}$ and $u_{(2)}$.

\par The similarity computation step of \algoname{elruna} consists of running Algorithm (\ref{algo:node_coverage}) and (\ref{algo:ini_solution}), which outputs the final similarity matrix $\mathbf{S}$.

\begin{algorithm}[htb!]
\DontPrintSemicolon
% ------------ Input -------------%
\KwIn{${G}_1 = ({V}_1, {E}_1)$ , ${G}_2 = ({V}_2, {E}_2)$, $t_{max}, \textbf{T1}, \textbf{T2}$}

% ------------ Output -------------%
\KwOut{The similarity matrix $\textbf{S}$ }

% ----------- Main Iteration ---------------%
\For{$k \gets 1$ \textbf{to} $t_{max}$} {
    $\text{\bf{S}}^{(k)} \gets n_1 \times n_2$ similarity matrix\;
    $\text{\bf{b1}}^{(k)} \gets n_1 \times 1$ vector with all entries equal to $-1$ \;
    
    $\text{\bf{b2}}^{(k)} \gets n_2 \times 1$ vector with all entries equal to $-1$ \;
    
    Update $\mathbf{c1}^{(k-1)}$ and $\mathbf{c2}^{(k-1)}$ based on Equation \ref{eq:ct} \;
    
    \For{$i$ \textbf{in} ${V}_1$} {
        \For{$u$ \textbf{in} ${V}_2$}{
            $\mathbf{e} \gets $ empty associative array \;
            $sum \gets$ $0$ \;
            \For{$j$ \textbf{in} ${N}(i)$}{
                \For{v \textbf{in} ${N}(u)$}{
                 \If{$\mathbf{S}^{(k-1)}_{jv} \geq \min \{ \mathbf{c1}^{(k-1)}_j, \; \mathbf{c2}^{(k-1)}_v\}$}{
                        $\mathbf{e}$[($j$, $v$)] $\gets$ $\mathbf{S}^{(k-1)}_{jv}$ \;
                    }
                }
            }
            $\mathbf{e} \gets $ sort by \textbf{\textit{value}} in descending order \;
            
            \For{$(j,v)$ \textbf{in} $\mathbf{e}$.\textbf{keys}}{
                \If{\text{$j$ and $v$ are not selected}}{
                    Accumulate $sum$ based on Rule 2 \;
                    Mark $j$ and $v$ as selected \;
                }
            }
            $\mathbf{S}^{(k)}_{iu} \gets \frac{sum}{\max \{ \sum_{j \in {N}(i)} \mathbf{b1}^{(k-1)}_j, \sum_{v \in {N}(u)} \mathbf{b2}^{(k-1)}_v \}}$ \;
            
            \If{$\mathbf{S}^{(k)}_{iu} > \mathbf{b1}^{(k)}_i$}{
                $\mathbf{b1}^{(k)}_i = \mathbf{S}^{(k)}_{iu}$
            }
            \If{$\mathbf{S}^{(k)}_{iu} > \mathbf{b2}^{(k)}_u$}{
                $\mathbf{b2}^{(k)}_u = \mathbf{S}^{(k)}_{iu}$
            }
        }
    }
}
\Return{\textbf{S}}
\caption{Similarity Computation}
\label{algo:ini_solution}
\end{algorithm}

We have added pictorial examples of Rule 1 and Rule 3 in the Appendix  for a better demonstration of the algorithm.

\subsection{Step 2: Building Alignment of Vertices}
After obtaining the final cross-network similarity matrix $\mathbf{S}$, we use two methods to extract the mappings between vertices from two networks:  \texttt{naive} and \texttt{seed-and-extend} alignment.

%\begin{figure} [!h]
%\vspace{-1cm}
\begin{wrapfigure}{r}{0.47\textwidth}
    \centering
    \includegraphics[width=\linewidth]{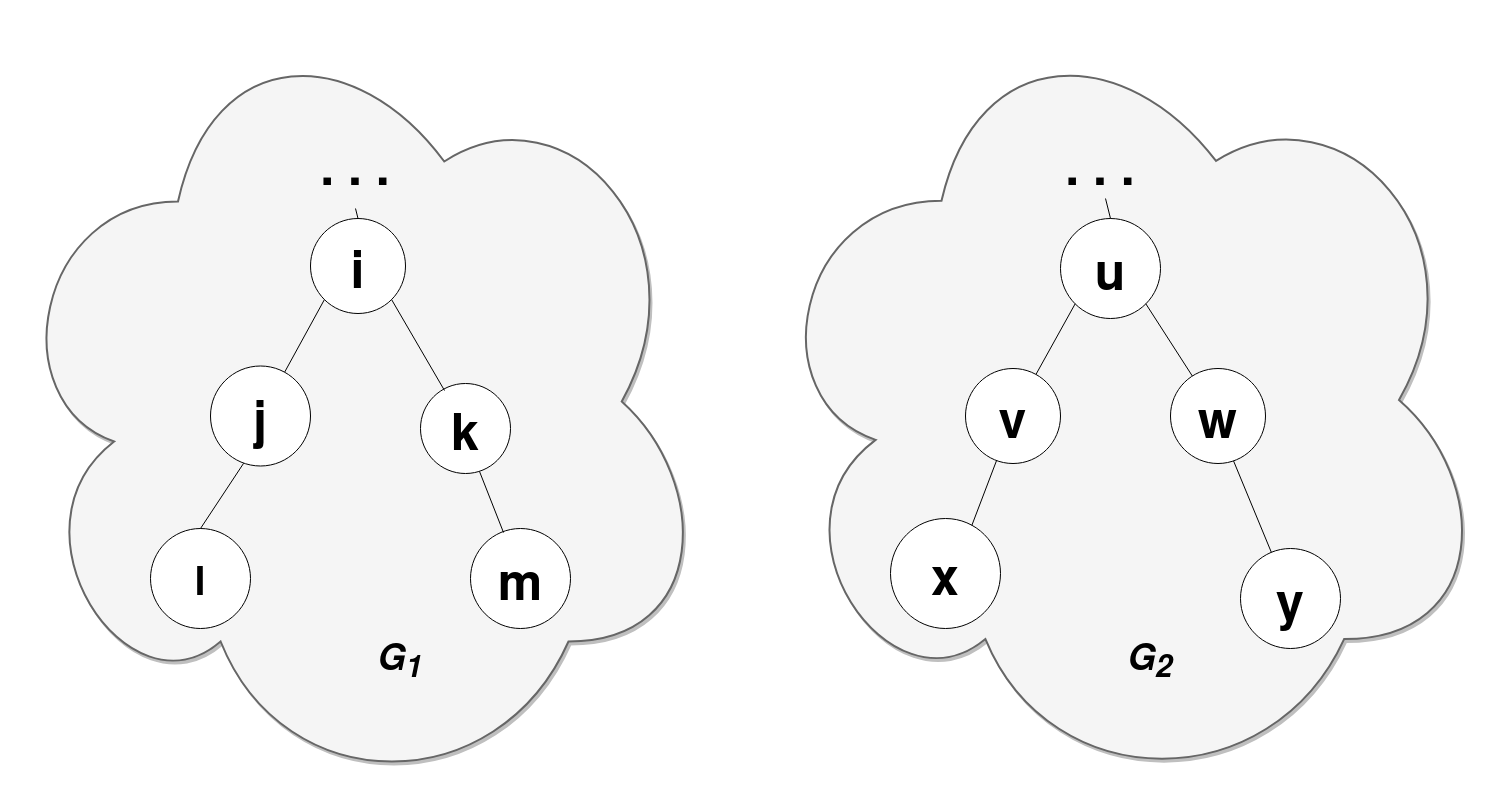}
    \caption{Failure of naive alignment  to distinguish symmetric nodes}
    \label{fig:naive_fail}
\end{wrapfigure}
%\end{figure}

\subsubsection{Naive Alignment}
Following the literature~\cite{final}, we sort all pairs of cross-network vertices by similarities in descending order. Then we iteratively align the next pair of unaligned vertices until all nodes in the smaller network are aligned. We note that this is a relatively simple alignment method, while many other algorithms~\cite{natali, netal, feizi2019spectral, modulealign, hubalign, netalign, ghost} use more complicated and computationally expensive alignment methods. As shown in the experimental section, however, using this naive alignment method, \algoname{elruna} significantly outperforms other baselines.

While the \texttt{naive} alignment method produces alignments of competitive quality, we observe that it fails to distinguish nodes that are topologically symmetric. An example is shown in  Figure (\ref{fig:naive_fail})  where ${G}_1$ and ${G}_2$ are isomorphic. In this circumstance, $k_{(1)}$ is equally similar to $v_{(2)}$ and $w_{(2)}$. At the same time, $m_{(1)}$ is equally similar to $x_{(2)}$ and $y_{(2)}$. If we break ties randomly,  it is possible that $k_{(1)}$ and $m_{(1)}$ may be mapped to vertices on different branches of the tree,  causing the loss of the number of conserved edges.

\subsubsection{Seed-and-Extend Alignment}
During the alignment process, nodes that have been aligned can serve as guidance for aligning other nodes~\cite{netal}. Referring back to  Figure (\ref{fig:naive_fail}), aligning $k_{(1)}$ to $w_{(2)}$ should imply that $m_{(1)}$ ought to be aligned with $y_{(2)}$ rather than $x_{(2)}$. To address the limitation of the \texttt{naive} alignment methods, we iteratively find the pair of \textbf{unaligned} nodes with the highest similarity. Then we align them and increase similarities between every pair of their \textbf{unaligned} cross-network neighbors by some small constant. Iterations proceed until all nodes in the smaller network are aligned. For efficiency, we use a red-black tree to store pairs of nodes. 

\subsection{Time Complexity of \algoname{elruna}}
Without loss of generality, we assume that two networks have a comparable number of vertices and edges. Let $n$ and $m$ denote the number of vertices and edges, respectively. Let $t_{max}$ denote the total number of iterations. 
One can easily see  that Algorithm (\ref{algo:node_coverage}) runs in $O(n^2 + mn)$ time and the naive alignment method runs in $O(n^2 \log n)$ time.
\begin{lemma}
The time complexity of Algorithm  (\ref{algo:ini_solution}) is $O(t_{max}m^2 \log n)$.
\end{lemma}

\begin{proof}
Let $d_i$ denote the degree of vertex $i_{(1)}$. Operations on lines $12$ to $13$ and lines $19$ to $23$ take constant time; thus the nested for loops on lines $10$ to $13$ take $\Theta(d_i d_u)$ time for every pair of $i_{(1)}$ and $u_{(2)}$. Sorting on line $14$ takes $\Theta(d_i d_u \log d_i d_u)$ time, and the for loop on lines $15$ to $18$ takes $\Theta(d_i d_u)$ time. We observe that $O(\log d_id_u) = O(\log n^2) =  O(\log n)$; therefore, the outer nested for loop on lines $6$ to $23$ takes
\begin{align}
    O(\sum_{i \in {V}_1} \sum_{u \in {V}_2} d_i d_u \log(d_i d_u)) &= O(\log n \sum_{i \in {V}_1} d_i \sum_{u \in {V}_2} d_u) \\
    &= O(m^2 \log n). \nonumber
\end{align}
The time complexity of the Algorithm (\ref{algo:ini_solution}) is $O(t_{max} \cdot m^2 \log n)$. \end{proof}

\begin{lemma}
The time complexity of the \texttt{seed-and-extend} alignment method is $O((n^2 + mn) \log(n^2 + mn))$.
\end{lemma}

\begin{proof}
At the beginning of the algorithm, adding all $n^2$ pairs of nodes into the red-black tree takes $O(n^2 \log n)$ time. Whenever we align a pair of vertices ($i_{(1)}, u_{(2)}$), we increase the similarities between all pairs of their unaligned cross-network neighbors. To update the corresponding similarities in the red-black tree, we add those pairs with new similarities into the tree. The total number of insertions of such pairs is 

$$
    \sum_{i_{(1)} \in {V}_1} d_i d_{f(i)},
$$
 where $f(i) \in {V}_2$ is the node that $i_{(1)}$ is aligned to. One can easily see that the function above is upper bounded by $O(mn)$, which implies that the total number of elements in the tree is $O(n^2 + mn)$. We perform at most $n^2 + mn$ number of \textit{find\_max}, \textit{deletion}, and \textit{insertion} operations; therefore, the overall running time of the algorithm is $O((n^2 + mn) \log(n^2 + mn))$. 
\end{proof}

The overall running time of the proposed \algoname{elruna} is $O(t_{max} \cdot m^2 \log n)$. If we assume that $O(m) = O(n \log n)$, which is a fair assumption for the networks that are used in the experiments, then the overall running time is $O(t_{max} \cdot n^2 \log^3 n)$.

\section{\algoname{rawsem} : Random-Walk-Based Selection Method}
In this section, we introduce the proposed selection rule \algoname{rawsem} for our local search procedure. We first discuss the baseline, which is used as the comparison method with \algoname{rawsem}. We then present the mechanism of \algoname{rawsem}.

\subsection{The Baseline}
The aim of the baseline is to further improve the objective value. Following the literature~\cite{shaydulin2018community}, given the permutation matrix produced by any alignment algorithm, the baseline algorithm constructs the search space for local search by selecting a subset of vertices from the smaller network and generating all permutations of their alignments while fixing the alignment of all other vertices that are not in the subset~\cite{shaydulin2018community} For each of the permutation of the alignment between selected vertices, the baseline selects the one that improves the objective the most. The baseline iterates until the objective has not been improved for a fixed number of iterations. 

\par Given ${G}_1 = ({V}_1, {E}_1)$ and ${G}_2 = ({V}_2, {E}_2)$, we transform the alignment matrix $\textbf{P}$ to the $n_1 \times 1$ alignment vector $\widetilde{\Pi}$ for which $\widetilde{\Pi}_i = u_{(2)}$ implies that vertex $i_{(1)}$ is aligned to vertex $u_{(2)}$.  
\par For each iteration, we randomly select a subset of vertices ${V}_1' \subset {V}_1$ with a fixed cardinality. Let ${V}_2' = \{ \widetilde{\Pi}_{i} : i \in {V}_1' \}$. Let $\mathbf{A}$ and $\mathbf{B}$ be the adjacency matrices of $G_1$ and $G_2$, respectively. The baseline local search explores all feasible solutions in the search space and attempts to find a new alignment vector $\Pi$ with a lower objective:
\begin{align}
    \min_{\Pi} & \{ - \frac{1}{2} \sum_{i,j \in {V}_1'} (\textbf{A}_{ij} \textbf{B}_{\Pi_i \Pi_j} + \sum_{k \in {N}(i) | k \notin {V}_1'} \textbf{A}_{ik} \textbf{B}_{\Pi_i \widetilde{\Pi}_k} + \\ \nonumber
    & \sum_{p \in {N}(j) | p \notin {V}_1'} \textbf{A}_{jp} \textbf{B}_{\Pi_j \widetilde{\Pi}_p}) \}\\ \nonumber
    & \text{s.t.} \;\; \Pi_i \in {V}_2' \;\; \forall i \in {V}_1' .\\ \nonumber
\end{align}

The baseline local search proceeds until the optimum has been reached, meaning that the objective has not been improved in over a number of iterations. 

\subsection{\algoname{rawsem} Algorithm}
\par Depending on the quality of the initial solution, it is possible that only a small fraction of vertices are not mapped optimally. Therefore, the baseline approach, which constructs the subset ${V}'_1$ with random selections from the entire vertex set, is not efficient. Ideally, we want to locate vertices that are not mapped optimally and only permute the alignments between them.

\par For simplicity, suppose nodes in ${V}_1$ are labeled with consecutive integers starting from $1$ and that nodes in ${V}_2$ are labeled with consecutive integers starting from $|{V}_1| + 1$. To quantify the level of mismatching of each vertex, we first define the concept of \textbf{violation}.

\begin{definition} [\textbf{Violation}]
 Let $\mathbf{o'}$ be an $n_1 \times n_2$ by $1$ vector. The \textbf{\textit{violation}} value of a vertex $i_{(1)} \in {G}_1$ is defined as 
 \begin{align} \label{eq:vio_i}
    \mathbf{o}'_i &= \sum_{j \in {N}(i)} (1 - \mathbf{B}_{\widetilde{\Pi}_i, \widetilde{\Pi}_j}) \\ &= |{N}(i)| - \sum_{j \in {N}(i)} \mathbf{B}_{\widetilde{\Pi}_i, \widetilde{\Pi}_j}. \nonumber
\end{align}
Informally, $\mathbf{o}'_i$ is the number of the neighbors of $i_{(1)}$ that are not conserved by $i_{(1)}$.

Let $\Bar{{V}}_2 = \{ u_{(2)} \in {V}_2 \; : \; \widetilde{\Pi}_i = u_{(2)} \; \exists \; i_{(1)} \in {V}_1 \}$ be the subset of ${V}_2$ consisting of all aligned vertices in $G_2$. Let $\widetilde{\Pi}^{-1} : \Bar{{V}}_2 \rightarrow {V}_1$ be the inverse mapping. Then the violation value of vertex $u \in {V}_2$ is defined in the same fashion:
\begin{align} \label{eq:vio_u}
    \mathbf{o}'_u & = \sum_{v \in {N}(u)} (1 - \mathbf{A}_{\widetilde{\Pi}^{-1}_u, \widetilde{\Pi}^{-1}_v})\\  &= |{N}(u)| - \sum_{v \in {N}(u)} \mathbf{A}_{\widetilde{\Pi}^{-1}_u, \widetilde{\Pi}^{-1}_v}. \nonumber
\end{align}

\end{definition}

We normalize violations of vertices by their degrees. Let $\mathbf{o}$ be a vector that encodes the normalized violation value of each vertex. Then we have

\begin{equation}
    \mathbf{o} = \mathbf{D}^{-1}\mathbf{o'},
\end{equation}
where $\mathbf{D}$ is the diagonal degree matrix such that $\mathbf{D}_{i,i}$ is the degree of vertex $i$. $\mathbf{o}$ provides the \textit{initial level of mismatching} of each vertex. We further normalize $\mathbf{o}$ such that its $L1$ norm equals  $1$.

\par From a high level, \algoname{rawsem} is a two-step procedure.

\begin{enumerate}
    \item[] (1) \textbf{Ranking vertices}: Starting from the initial level of mismatching of vertices, \algoname{rawsem} iteratively updates $\mathbf{o}$ in a \texttt{PageRank}~\cite{pagerank} fashion. Then, it ranks vertices by the  converged levels of mismatching.
    \item[] (2) \textbf{Local search}: The search space is constructed based on the ranking, and  the local search is performed.
\end{enumerate}

\subsubsection{Step 1 -- Ranking Vertices}
\begin{wrapfigure}{r}{0.35\textwidth}
    \centering
    \includegraphics[width=\linewidth]{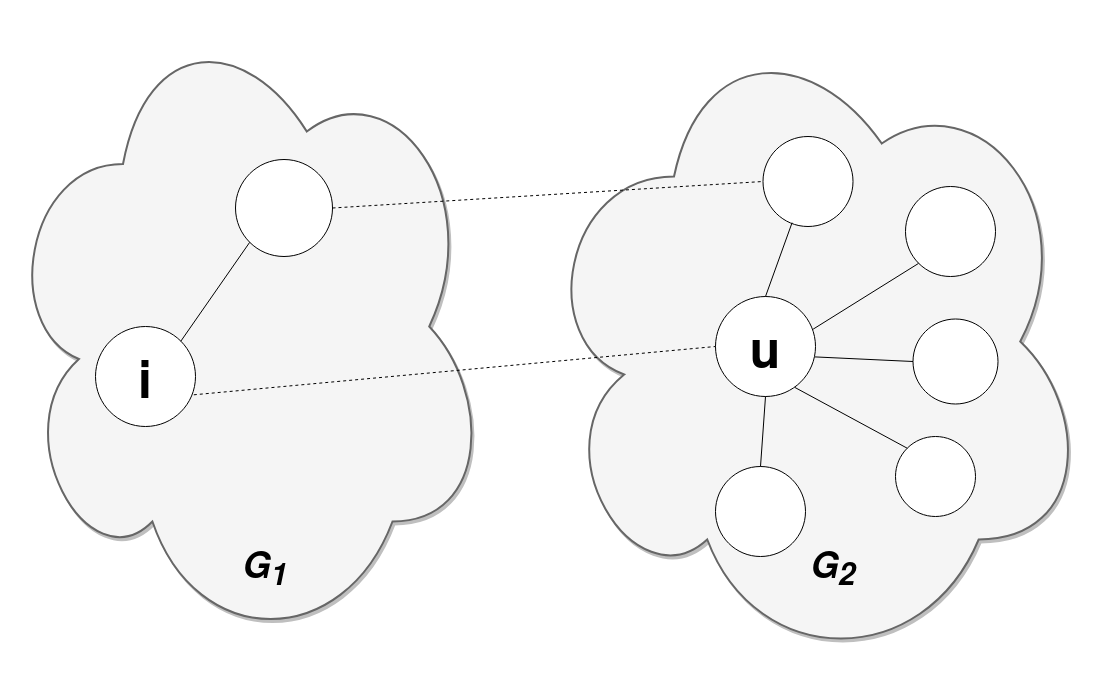}
    \caption{Zero violation of vertex $i$}
    \label{fig:zero_violation}
\end{wrapfigure}

For two real-world networks where isomorphism does not exist, we expect many vertices to have nonzero initial violations. Clearly, a nonzero violation does not imply a non-optimal mapping. We note that a zero violation does not always imply an optimal mapping. As shown in the Figure (\ref{fig:zero_violation}) where dashed lines indicate mappings, $i_{(1)}$ has a zero violation. This suggests the insufficiency of initial violation values. 

\begin{figure}[!h]
\centering
\includegraphics[width=1\linewidth]{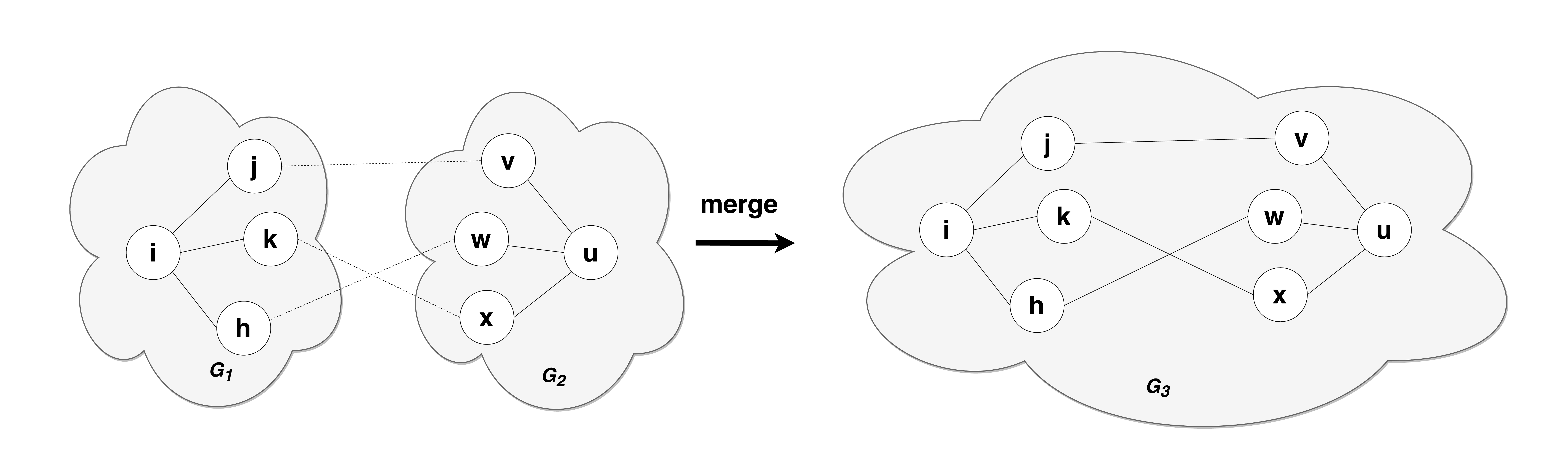}
\caption{Example of the merge operation}
\label{fig:merge}
\end{figure}

\par To adapt the network information into our model, we use an iterative approach based on the intuition that \emph{a vertex is more mismatched if its neighbors are more mismatched}. One immediate method is to propagate a violation value via edges throughout the networks. Let ${G}_2' = (\Bar{{V}}_2, \Bar{{E}}_2)$ be the subgraph induced by $\Bar{{V}}_2$ (as stated above, $\Bar{{V}}_2$ is the set of aligned vertices in $G_2$). To start with, we merge ${G}_1$ and ${G}_2'$ by adding edges to connect aligned cross-network vertices. Figure (\ref{fig:merge}) illustrates a example of the merge operation. We denote the newly constructed undirected graph ${G}_3 = ({V}_3, {E}_3)$, where ${V}_3 = {V}_1 \cup \Bar{{V}}_2$ and ${E}_3 = {E}_1 \cup \Bar{{E}}_2 \cup \{ (i_{(1)}, u_{(2)})  :  i_{(1)} \in {V}_1, u_{(2)} =  \widetilde{\Pi}_i\}$. Let $\textbf{C}$ denote the adjacency matrix of ${G}_3$. Let $\textbf{R}$ be the vector that encodes the propagated levels of mismatching for each vertex. Let \textbf{D} denote the diagonal degree matrix such that $\textbf{D}_{i,i} = |{N}(i)|$. We propagate violations in a \texttt{PageRank}~\cite{pagerank} fashion:

\begin{equation}
  \textbf{R}^{(k)}_i = \alpha \sum_{j} \frac{\textbf{C}_{ij}}{\textbf{D}_{j,j}} \textbf{R}^{(k-1)}_j + (1 - \alpha) \textbf{o}_i. 
\end{equation}

In matrix notation, this is

\begin{equation} \label{eq:pr}
    \textbf{R}^{(k)} = \alpha \textbf{CD}^{-1}\textbf{R}^{(k-1)} + (1-\alpha)\textbf{o} .
\end{equation}

By initializing $\mathbf{R}^{(0)}$ as a probability vector, we can rewrite Equation (\ref{eq:pr}) to

\begin{equation} \label{eq:pr_final}
    \textbf{R}^{(k)} = [\alpha \textbf{CD}^{-1} + (1 - \alpha ) \textbf{o} \textbf{1}^T] \; \textbf{R}^{(k-1)},
\end{equation}
where $\mathbf{1}$ is the vector with all entries equal to $1$.

Equation (\ref{eq:pr_final}) encodes an eigenvalue problem that can be approximated by power iteration. Let $\textbf{E} = \alpha \textbf{CD}^{-1} + (1 - \alpha ) \textbf{o} \textbf{1}^T$. By the undirected nature of ${G}_3$, the transition matrix $\textbf{C}\textbf{D}^{-1}$ is \textit{irreducible}; therefore $\textbf{E}$ is a left-stochastic matrix with a leading eigenvalue equal to $1$, and the solution of Equation (\ref{eq:pr_final}) is the principal eigenvector of $\textbf{E}$. Moreover, $\textbf{\textbf{E}}$ is also primitive; therefore the leading eigenvalue of $\textbf{E}$ is unique, and the corresponding principal eigenvector can be chosen to be strictly positive. As a result, the power iteration converges to its principal eigenvector. 

\par The violation vector $\textbf{o}$ plays the role of teleportation distribution, which encodes external influences on the importance of vertices. The converged $\textbf{R}$ gives the levels of mismatching of vertices, and a vertex with a higher value is even more mismatched. We rank vertices in $G_1$ based on $\mathbf{R}$ in descending order.

\subsubsection{Step 2 -- Local Search}
Based on the ranking produced in the previous step, \algoname{rawsem} uses a sliding window over sorted vertices to narrow down the search space. Let $m$ be the size of the window with the tail lying at the highest-ranked vertices. For each iteration, we construct ${V}'_1$ by randomly selecting vertices within the window and perform the local search. If the objective has not been improved in $s$ iterations, we move the window $l$ nodes forward, then continue the local search procedure. The local search  terminates when the objective has not been improved for $s_{max}$ number of iterations. In our experiments, we set $|{V}_1| = 6$.

% -----------------------------------
% -       Experiment Results        -
% -----------------------------------
\section{Experimental Results}
\par In this section, we first present the experimental setup and performance of the proposed \algoname{elruna} in comparison with $8$ baseline methods over three alignment scenarios. Then, we study the time-quality trade-off and the scalability of \algoname{elruna}. \emph{We emphasize that \algoname{elruna} is not a local search method and should not be compared with local search methods because \algoname{elruna} can serve as a preprocessing step to all local searches.} We demonstrate that \algoname{Rawsem} could further improve the alignment quality with significantly fewer iterations than  the naive local search method requires.

\par {\bf Reproducibility:} Our source code, documentation, and data sets are available at \url{https://tinyurl.com/uwn35an}.
\par \textbf{Baselines for \algoname{elruna}}. We compare \algoname{elruna} with 8 state-of-the-art network alignment algorithms: \texttt{IsoRank} \cite{isorank}, \texttt{Klau} \cite{natali}, \texttt{NetAlign} \cite{netalign}, \texttt{REGAL} \cite{heimann2018regal}, \texttt{EigenAlign} \cite{feizi2019spectral}, \texttt{C-GRAAL} \cite{memivsevic2012c}, \texttt{NETAL} \cite{netal}, and \texttt{HubAlign} \cite{hubalign}. These algorithms, published in years 2008--2019, have proven superior to many other methods, so we  choose them. Other methods that perform in significantly longer running time have not been considered.
The baseline methods are described in the Related Work section. \texttt{C-GRAAL}, \texttt{Klau} and \texttt{Netalign} require prior similarities between cross-network vertices as input. Following \cite{final, heimann2018regal}, we use the degree of similarities as the prior similarities. For \texttt{Klau} and \texttt{Netalign}, as suggested by \cite{heimann2018regal}, we construct the prior alignment matrix by choosing the highest $k \times \log_2 n$ vertices, where $k = 5$. Note that \textbf{\algoname{elruna} does not requires prior knowledge about the similarities between cross-network vertices}.

\par We do not compare \algoname{elruna} with \texttt{FINAL} \cite{final} because \texttt{FINAL} solves a different problem, namely, the attributed network alignment problem. \texttt{ModuleAlign} \cite{modulealign} is also not chosen as a baseline method because it is the same as \texttt{HubAlign} \cite{hubalign} except that \texttt{ModuleAlign} uses a different method to optimize the biological similarities between vertices. Additionally, we do not compare with $\texttt{GHOST}$ \cite{ghost} because its signature extraction step took hours even for small networks.

\par \textbf{Experimental Setup for \algoname{elruna}}. All experiments were performed on an Intel Xeon  E5-2670 machine with 64 GB of RAM. For the sake of iterative progress comparison, we set the maximum number of iterations $t_{max}$ to the larger diameter of the two networks.
\par Our experiment consists of three scenarios: (1) \textit{\textbf{self-alignment without and under the noise}}, (2) \textit{\textbf{alignment between homogeneous networks}}, and (3) \textit{\textbf{alignment between heterogeneous networks}}. Detailed descriptions of each category are presented   later. In general, the first test scenario \textit{self-alignment without and under the noise} consists of 12 networks from various domains. For each network, we generate up to 14 noisy copies with increasing noise levels (defined later) up to $25\%$. This gives us a total of $1,64$ pairs of a network to align. The second test scenario \textit{alignment between homogeneous networks} has $3$ pairs of networks for which each pair consists of two subnetworks of a larger network. In our third test scenario \textit{alignment between heterogeneous networks}, we align 5 pairs of networks where each pair consists of two networks from different domains. We note that the third scenario is usually used by attributed network alignment algorithms. Therefore, it is not exactly what we solve with our formulation, but  we demonstrate the results because it is an  important practical task. To the best of our knowledge, our experimental setup is the most comprehensive in terms of the combination of the number of baselines, the number of networks, the categories of testing cases, and the levels of noise applied. 

\par \textbf{Evaluation Metric.} To quantify the alignment quality, we use two well-known metrics: the \textit{edge correctness} (\textit{EC}) \cite{netal} and the \textit{symmetric substructure score} ($S^3$) \cite{saraph2014magna}. Let $f(V_1) = \{u \in V_2 : \mathbf{P}_{i,u} = 1, \exists i \in V_1 \}$, and $f(E_1) = |\{ (f(i_{(1)}), f(j_{(1)})) \in E_2 \; : \; (i_{(1)}, j_{(1)}) \in E_1 \}|$. That is, $f(V_1)$ is the set of vertices in $G_2$ that are aligned (note that since we assume $|V_1| \leq |V_2|$, some vertices in $G_2$ are left unaligned), and $f(E_1)$ is the set of edges in $G_2$ such that for each edge, the alignment of its incident vertices is adjacent in $G_1$. Then, we have 

% ---------- EC ----------- %
\begin{equation}
    EC = \frac{|f(E_1)|}{|E_1|}
\end{equation}

and 

% ---------- S3 ----------- %
\begin{equation}
    S^3 = \frac{|f(E_1)|}{|E_1| + |E(G_2[f(V_1)])| - |f(E_1)|},
\end{equation}

where $|E(G_2[f(V_1)])|$ is the number of edges in the subgraph of $G_2$ induced by $f(V_1)$. 

\subsection{Self-alignment without and under the Noise} 
In this experiment, we analyze how \algoname{elruna} performs with structure noise being added to the original network. Given a target network $G_1$, simply aligning $G_1$ with its random permutations is not a challenge for most of the existing state-of-the-art network alignment algorithms. A more interesting test case arises when we try to align the original networks $G_1$ with its \textbf{noisy} permutations $G_2$ such that $G_2$ is a copy of $G_1$ with additional edges being added \cite{netal, bigalign, final, yasar2018iterative, heimann2018hashalign}. This scenario is even more challenging when the number of noisy edges is large with respect to the number of edges in $G_1$. In addition, this perturbation approach  reflects many real-life network alignment task scenarios \cite{netalign, heimann2018hashalign, heimann2018regal}.

\par Given the network ${G}_1 = ({V}_1, {E}_1)$, a noisy permutation of ${G}_1$ with noise level $p$, denoted by ${G}^{(p)}_2 = ({V}_2, {E}_2)$, is created with two steps:
\begin{enumerate}
    \item Permute ${G}_1$ with some randomly generate permutation matrix.
    \item Add $p|{E}_1|$ edges to ${G}_1$ uniformly at random by randomly connecting nonadjacent pairs of vertices. 
\end{enumerate}

Properties of each network $G_1$ are given in Table \ref{tab:self}. Note that under this model, the highest EC any algorithm can achieve by aligning ${G}_1$ and ${G}_2$ is always $1$. \emph{In this experiment, we demonstrate the results {\bf without} applying local search; that is,  we demonstrate how only \algoname{elruna} outperforms the  state-of-the-art methods.} 

\begin{center}
\captionof{table}{Networks for test case: \textit{self-align without and under noise}} \label{tab:self} 
 \begin{tabular}{||l c c c||} 
 \hline
 \textbf{Domain} & $n$ & $m$ & \textbf{label}\\ [0.5ex] 
 \hline 
 Barabasi random network & 400 & 2,751 & \texttt{barabasi} \\ 
 \hline
 
 Homle random network& 400 & 2,732 & \texttt{homel} \\ \hline
 
 Coauthorships & 379 & 914 & \texttt{co-auth\_1} \\ \hline
 
 Gene functional association & 993 & 1,300 & \texttt{bio\_1} \\ \hline

 Economy & 1,258 & 7,513 & \texttt{econ} \\ \hline
 
 Router & 2,113 & 6,632 & \texttt{router} \\ \hline
 
 Protein-protein interaction & 2,831 & 4,562 & \texttt{bio\_2} \\ \hline
 
 Twitter & 4,171 & 7,059 & \texttt{retweet\_1} \\ \hline
 
 Erods collaboration & 5,019 & 7,536 & \texttt{erdos} \\ \hline
 
 Twitter & 7,252 & 8,061 & \texttt{retweet\_2} \\ \hline
 
 Social interaction & 10,680 & 24,316 & \texttt{social} \\ \hline
 
 Google+ & 23,628 & 39,242 & \texttt{google+} \\ \hline

\end{tabular}

\end{center}

\par To demonstrate the effectiveness of \algoname{elruna}, we first solve the alignment problems on random networks generated by Bar\`abasi--Albert preferential attachment  model (BA model) \cite{barabasi} and the Holme-Kim model (HK model) \cite{holme}; see Table \ref{tab:self}. The HK model reinforces the BA model with an additional probability $q$ of creating a triangle after connecting a new node to an  existing node. In our experiment we set $q=0.4$. For each of the random networks, we generate $12$ noisy permutations with increasing noise level $p$ from 0 to 0.21. We then align ${G}_1$ with each of its noisy permutations using  \algoname{elruna} and the baseline algorithms. \algoname{elruna} has two versions, \texttt{\algoname{elruna}\_Naive} and \texttt{\algoname{elruna}\_Seed}, which differ by the alignment methods we use. The results are summarized in Figures \ref{fig:barabasi} and \ref{fig:homle}. 

\par Next, we solve the network alignment problem on 10 real-world networks \cite{nr-sigkdd16, leskovec2012learning, boguna2004models} from various domains. The details of the selected networks are shown in Table \ref{tab:self}. For each network, we use the same model to generate $14$ noisy permutations with increasing noise level $p$ from $0$ up to $0.25$. That is, $G_2^{(0.25)}$ has added additional $25\%$ noisy edges added to $G_1$. We then align $G_1$ with each of its permutation using  \algoname{elruna} and the baseline algorithms. The results are summarized in Figures \ref{fig:newman}--\ref{fig:google}. 

Similar to most existing state-of-the-art network alignment algorithms' testsets, the networks in our testset are sparse. This is one of the common limitations of the similarity matrix based approaches. Aligning dense graphs implies using dense similarity matrices (or comparable data structures that replace these matrices). Implementing scalable alignment algorithms for dense graphs that actually effectively manipulate dense similarity matrices requires a lot of effort (such as extremely fast matrix manipulations, parallelization, and specific dense data structures that are potentially different for different hardware) that is beyond the scope of this work because the focus will actually be more on the performance rather than the trade-off of quality and performance. To the best of our knowledge, the most competitive network alignment approaches are lack of effective implementation for dense graphs. 

\par \textbf{Results}. Clearly,  \algoname{elruna} significantly outperforms all baselines on all networks. In particular, both versions of \algoname{elruna} (\texttt{\algoname{elruna}\_Naive} and \texttt{\algoname{elruna}\_Seed}) outperform 6 existing methods---\texttt{REGAL}, \texttt{EigenAlign}, \texttt{Klau}, \texttt{IsoRank}, \texttt{C-GRAAL}, and \texttt{NetAlign}---by an order of magnitude under high noise levels. For the other two baselines, \texttt{NETAL} and \texttt{HubAlgin}, \algoname{elruna} also produces much better results than they do, with improvement up to $60\%$ under high noise levels. At the same time, both versions of \algoname{elruna} are robust to noise such that they output high-quality alignment even when $p$ reaches $0.25$. Noisy edges change the degrees of vertices by making them more uniformly distributed; in other words, they perform a process that can be viewed as network anonymization. This experiment demonstrates \texttt{ELRUNA}'s superiority in identifying the hidden isomorphism between two networks.

\par We observe that the alignment quality of \texttt{Klau} and \texttt{NetAlign} is similar. Such behaviors were also reported in other literature \cite{heimann2018regal, yasar2018iterative}. We also note that the trends of $EC$ and $S^3$ are almost the same. \texttt{EigenAlign} crashed on \texttt{bio\_1} (for $p > 0.17$), \texttt{bio\_2}, \texttt{econ} (for $p > 0.05$), \texttt{router}, and \texttt{erdos} networks, and it took over $23$ hours to even run on one instance of \texttt{social}, \texttt{google+}, and \texttt{retweet\_2} networks. \texttt{C-GRAAL} crashed on \texttt{econ} (ran successfully only on $p = 0.3$), \texttt{google+}, and \texttt{social} networks. \texttt{HubAlign} crashed on $\texttt{google+}$ networks.

% ---------- barabasi ----------- %
\begin{figure}[!h] 
    \centering
    \includegraphics[width=0.95\linewidth]{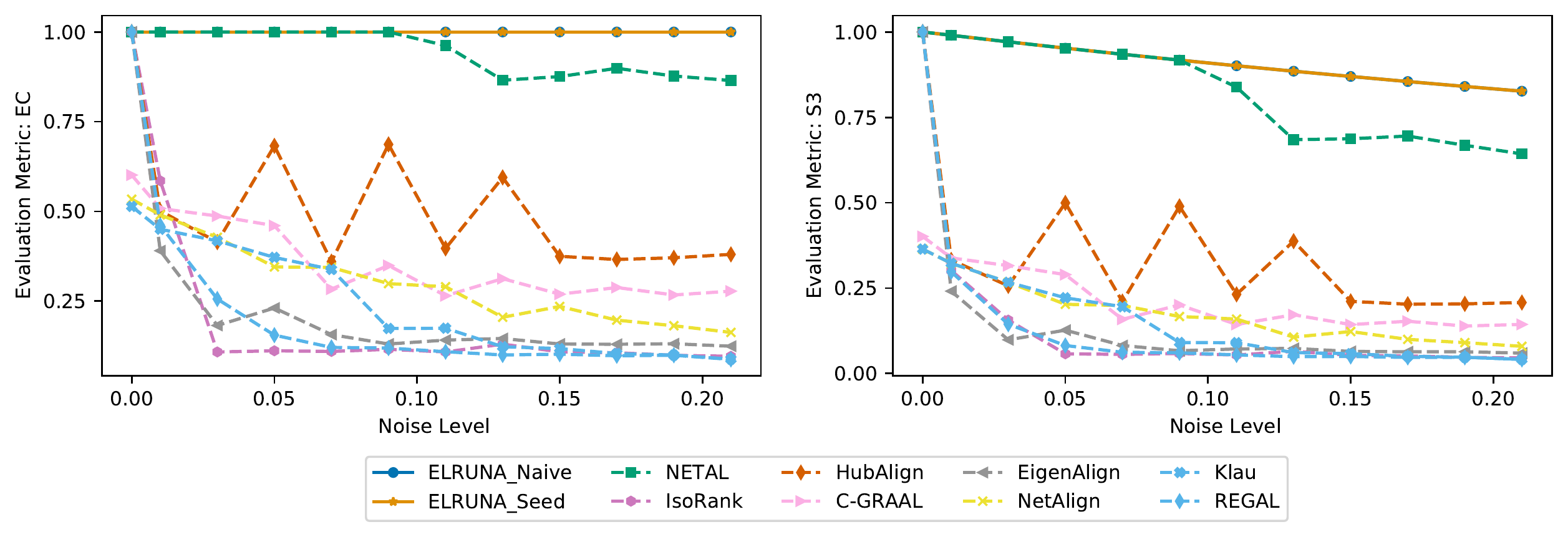}
    \caption{Alignment quality comparison on \texttt{barabasi} network}
    \label{fig:barabasi}
\end{figure}

% ---------- homle ----------- %
\begin{figure}[!h] 
    \centering
    \includegraphics[width=0.95\linewidth]{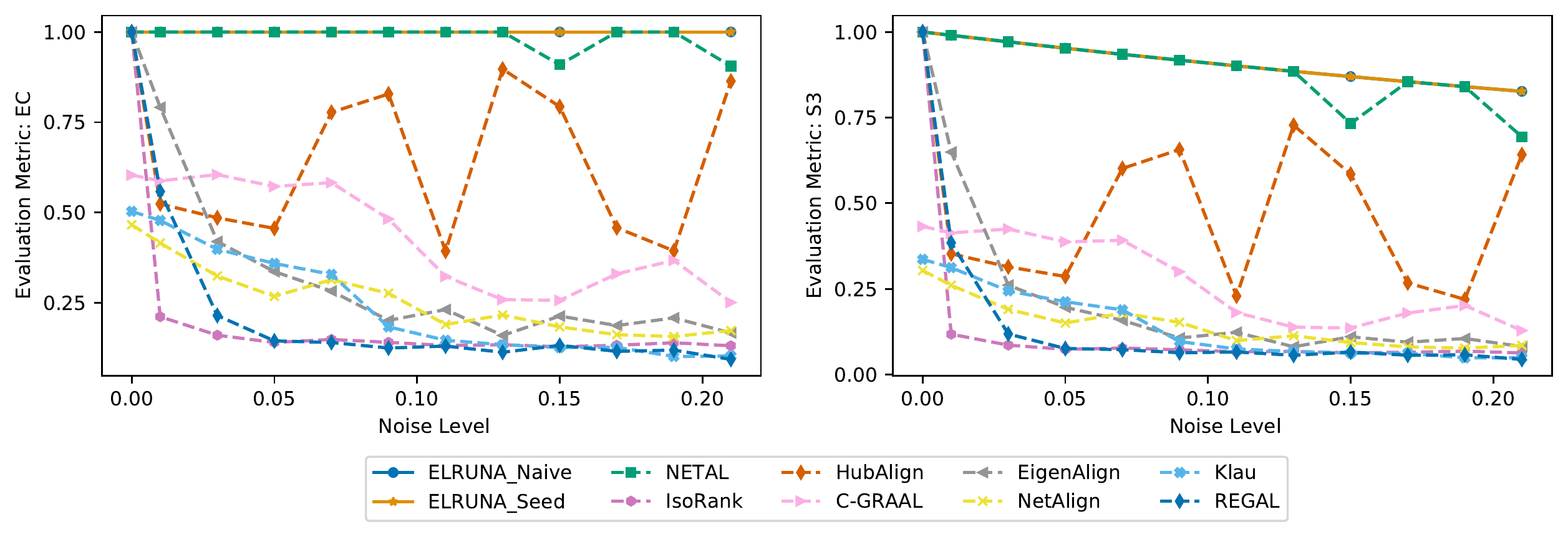}
    \caption{Alignment quality comparison on \texttt{homle} network}
    \label{fig:homle}
\end{figure}

% ---------- newman ----------- %
\begin{figure}[!h] 
    \centering
    \includegraphics[width=1\linewidth]{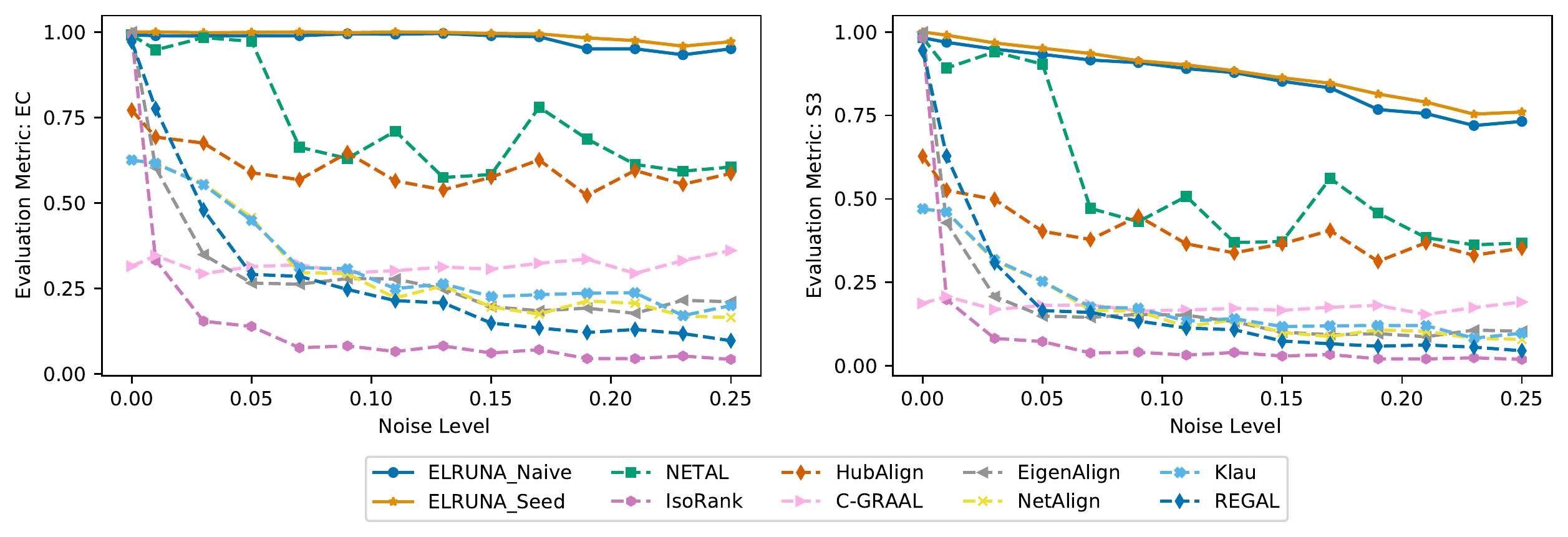}
    \caption{ Alignment quality comparison on \texttt{co-auth\_1} network}
    \label{fig:newman}
\end{figure}

% ---------- bio ----------- %
\begin{figure}[!h] 
    \centering
    \includegraphics[width=1\linewidth]{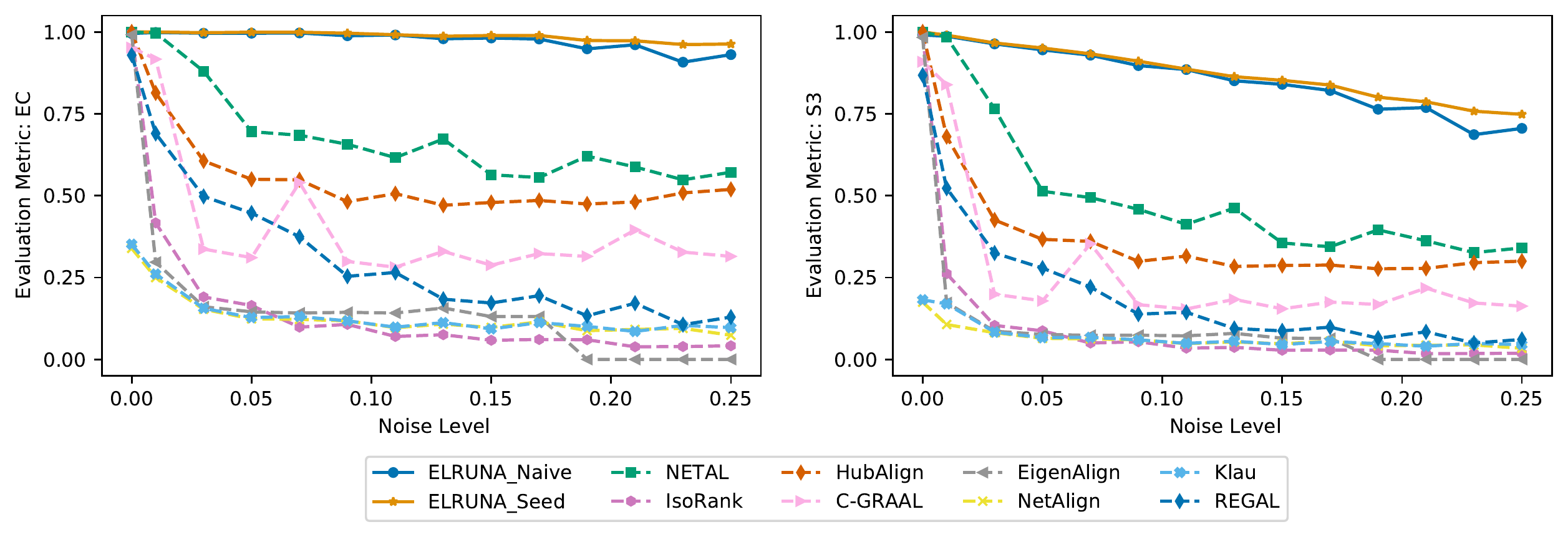}
    \caption{Alignment quality comparison on \texttt{bio\_1} network}
    \label{fig:bio}
\end{figure}

% ---------- econ ----------- %
\begin{figure}[!h] 
    \centering
    \includegraphics[width=1\linewidth]{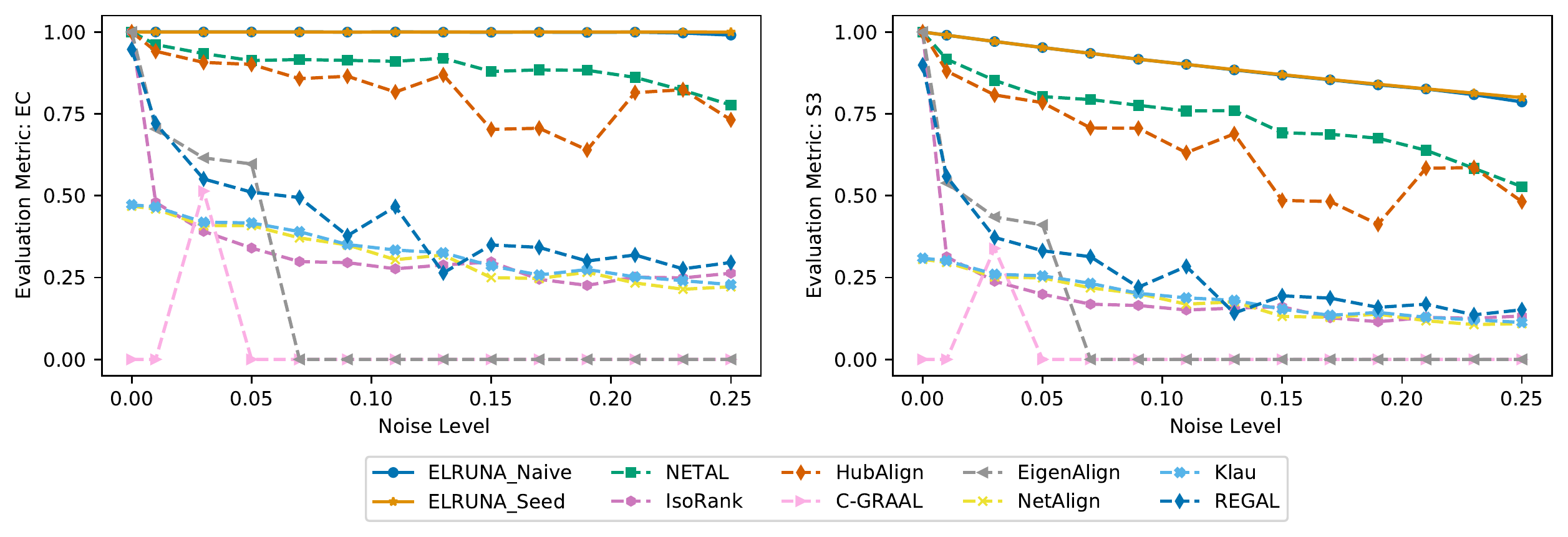}
    \caption{Alignment quality comparison on \texttt{econ} network}
    \label{fig:econ}
\end{figure}

\clearpage

% ---------- router ----------- %
\begin{figure}[!h] 
    \centering
    \includegraphics[width=1\linewidth]{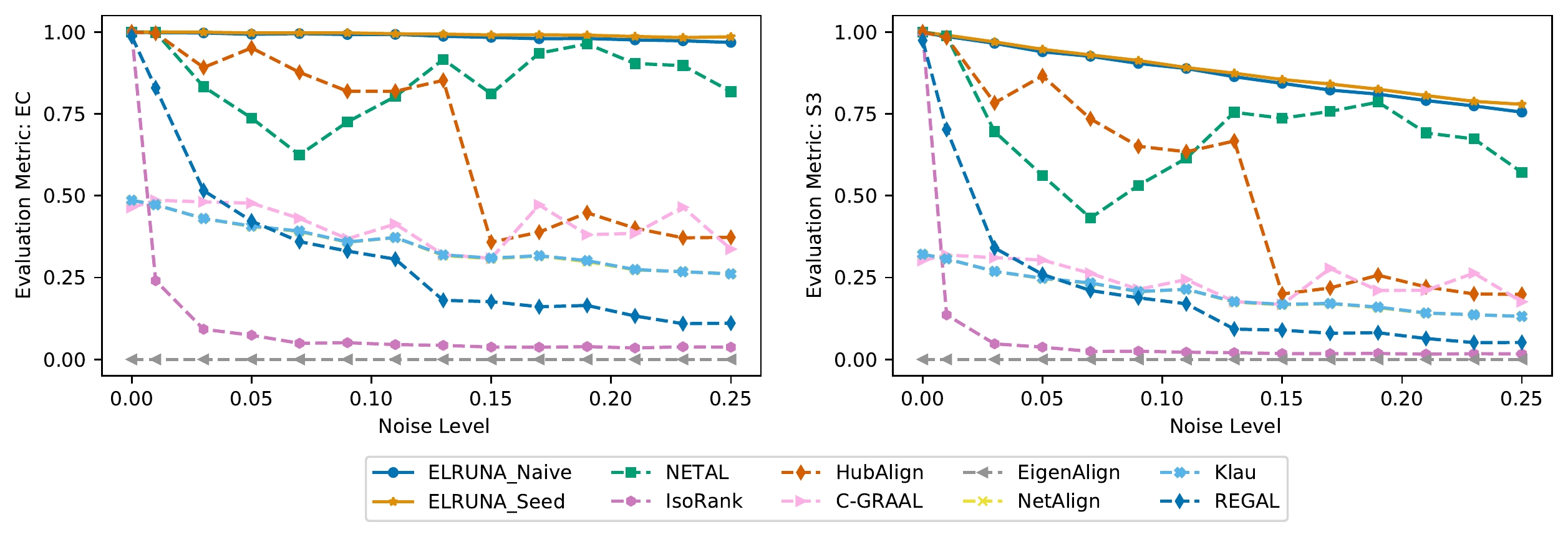}
    \caption{Alignment quality comparison on \texttt{router} network}
    \label{fig:router}
\end{figure}

% ---------- bio2 ----------- %
\begin{figure}[!h] 
    \centering
    \includegraphics[width=1\linewidth]{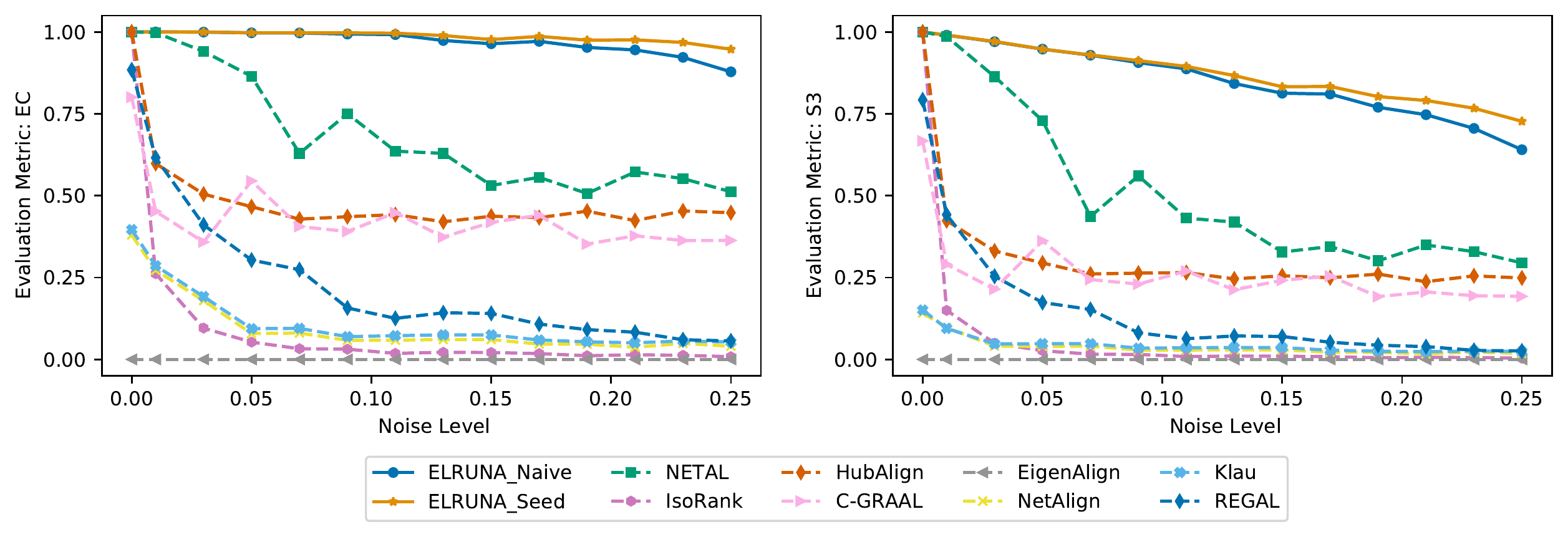}
    \caption{Alignment quality comparison on \texttt{bio\_2} network}
    \label{fig:bio2}
\end{figure}

% ---------- retweet ----------- %
\begin{figure}[!h] 
    \centering
    \includegraphics[width=1\linewidth]{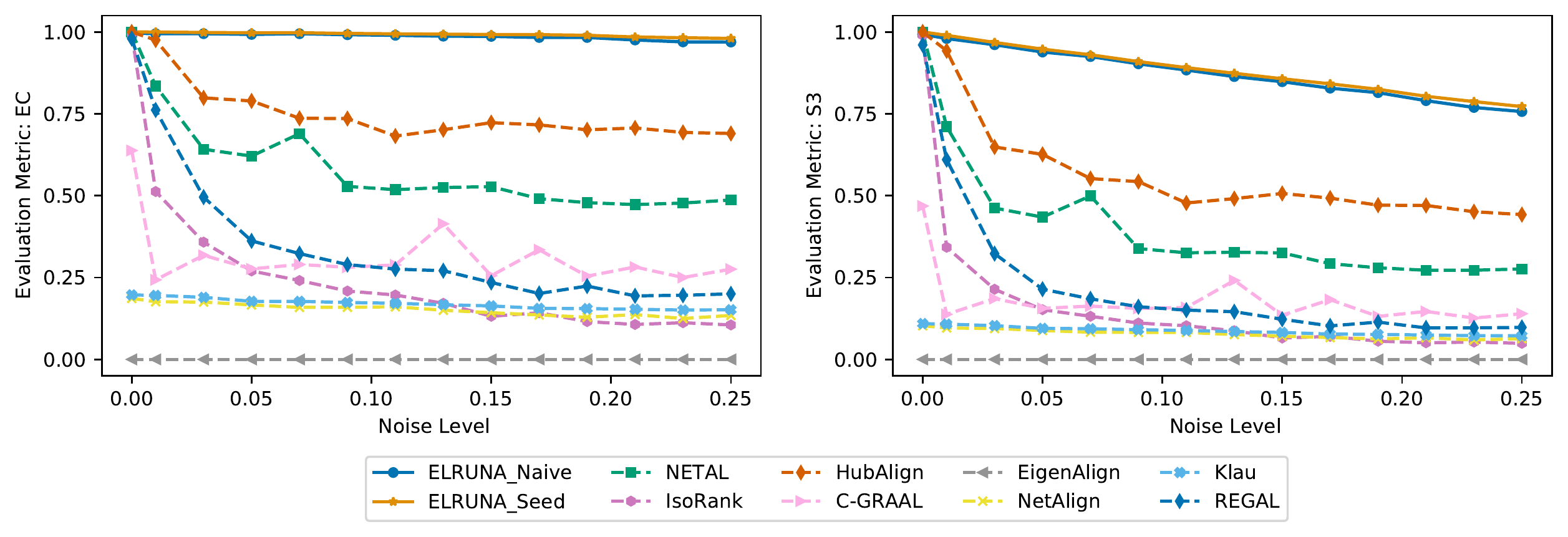}
    \caption{Alignment quality comparison on \texttt{retweet\_1} network}
    \label{fig:retweet}
\end{figure}

% ---------- erdos ----------- %
\begin{figure}[!h] 
    \centering
    \includegraphics[width=1\linewidth]{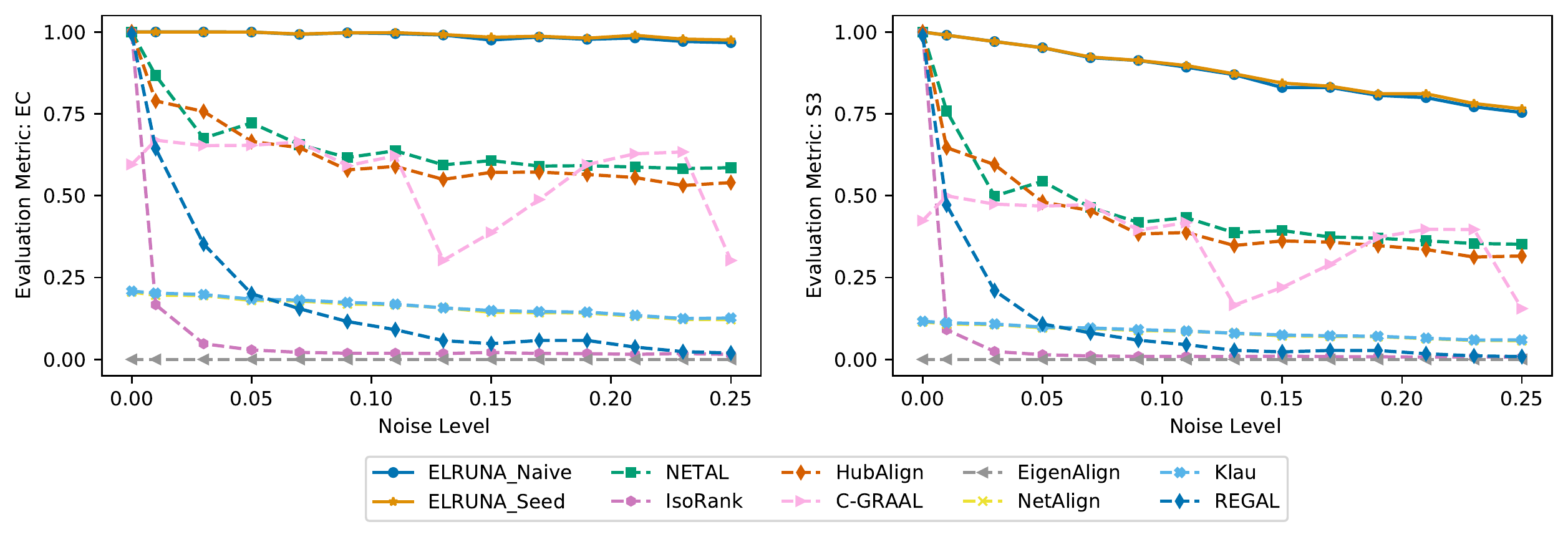}
    \caption{Alignment quality comparison on \texttt{erdos} network}
    \label{fig:erdos}
\end{figure}

% ---------- retweet2 ----------- %
\begin{figure}[!h] 
    \centering
    \includegraphics[width=1\linewidth]{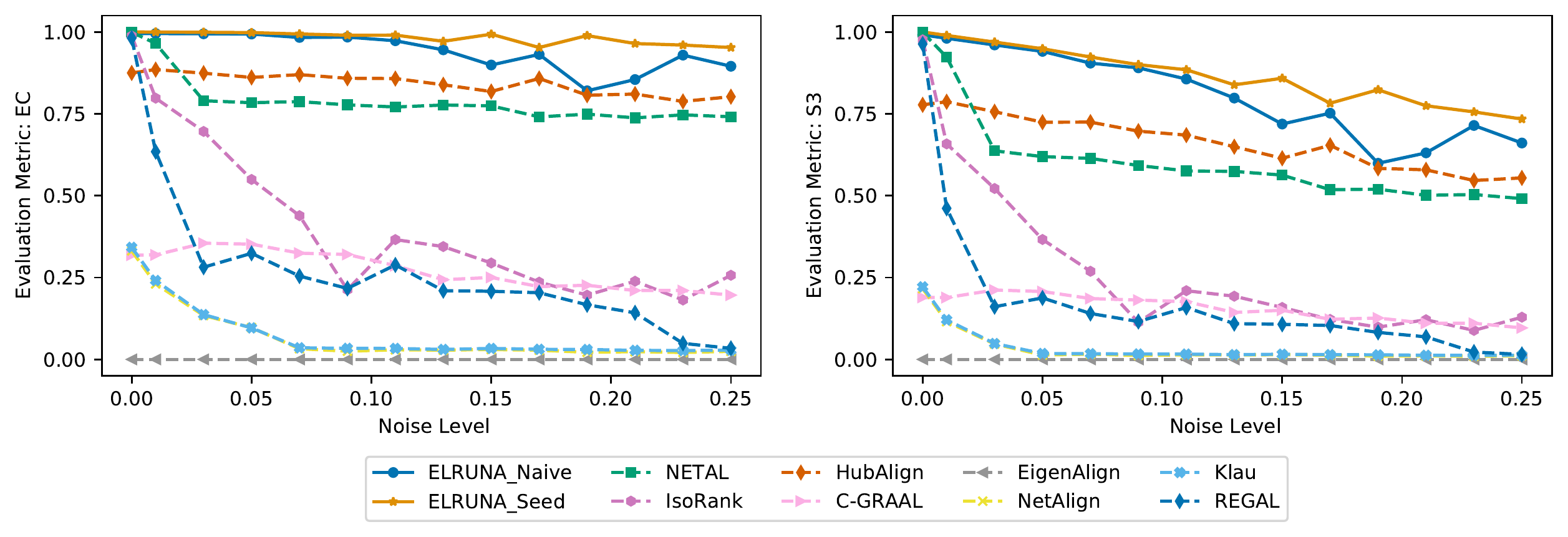}
    \caption{Alignment quality comparison on \texttt{retweet\_2} network}
    \label{fig:retweet2}
\end{figure}

% ---------- social ----------- %
\begin{figure}[!h] 
    \centering
    \includegraphics[width=1\linewidth]{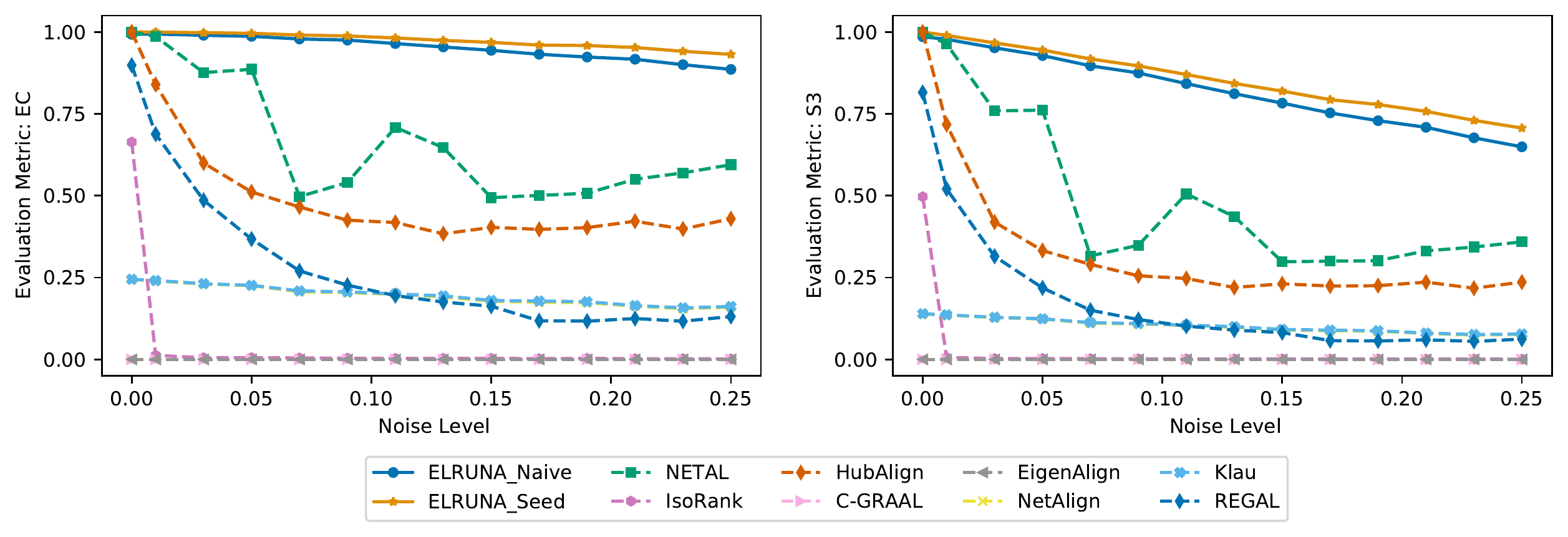}
    \caption{Alignment quality comparison on \texttt{social} network}
    \label{fig:social}
\end{figure}

\clearpage

% ---------- google ----------- %
\begin{figure}[!h] 
    \centering
    \includegraphics[width=1\linewidth]{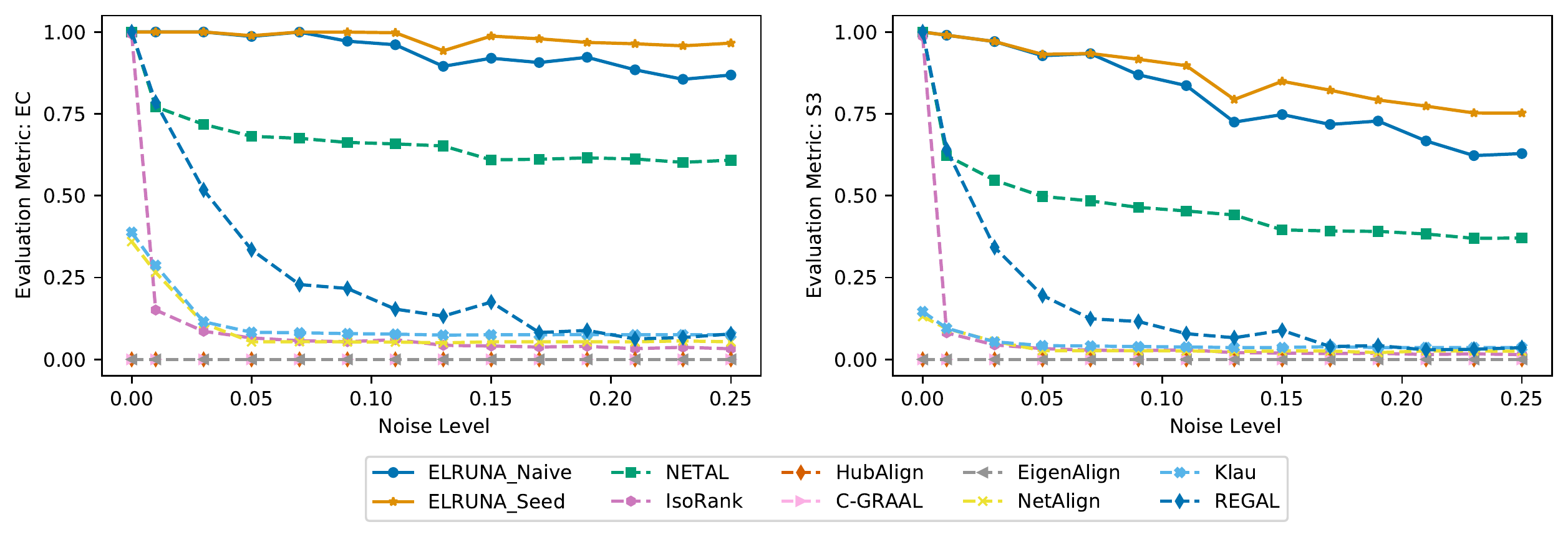}
    \caption{Alignment quality comparison on \texttt{google+} network}
    \label{fig:google}
\end{figure}

\subsection{Alignment between Homogeneous Networks}
In this experiment, we study how \algoname{elruna} performs when aligning two networks that were subgraphs of a larger network. Given a network $G$, we extract two \textit{induced} subnetworks $G_1$ and $G_2$ of $G$ for which $G_1$ and $G_2$ share a common set of vertices. We compare \texttt{\algoname{elruna}\_Naive} and \texttt{\algoname{elruna}\_Seed} with the state-of-the-art methods by aligning 3 pairs of $G_1$ and $G_2$ \cite{nr-sigkdd16, digg, facebook} from 3 domains, respectively. The properties of networks are listed in Table \ref{tab:homo}. Note that the implementation of $\texttt{REGAL}$ does not support alignment between networks with different sizes; therefore, it is not included in this testing case. In addition, the benchmark of \texttt{EigenAlign} is not included for the \texttt{facebook} networks because its running time was over 23 hours.

\par The first pair consists of two \texttt{DBLP} subnetworks with 2,455 overlapping nodes. The second pair of \texttt{Digg} social networks has 5,104 overlapping nodes, and  the \texttt{Facebook Friendship} network has 8,130 nodes in common. Note that under this setting, the highest $EC$ any algorithm can achieve is lower than 1. In fact, the optimal $EC$ value is not known. 

\begin{center}
\captionof{table}{Datasets for test case: \textit{Alignment between homogeneous networks}} \label{tab:homo} 
 \begin{tabular}{||c c c c||}
 \hline
 \textbf{Domain} & $n$ & $m$  & \textbf{label}\\ [0.5ex] 
 \hline \hline
 
 \texttt{DBLP} & 3,134 vs 3,875 & 7,829 vs 10,594 & \texttt{dblp} \\ \hline
 
 \texttt{Digg Social Network} & 6,634 vs 7,058 & 12,177 vs 14,896 & \texttt{digg} \\ \hline
  
 \texttt{Facebook Friendship} & 9,932 vs 10,380 & 26,156 vs 31,280 & \texttt{facebook} \\ \hline
 
\end{tabular}
\end{center}

\emph{\textbf{In this experiment, we demonstrate the results without local search. That is,  we demonstrate how \algoname{elruna} outperforms the  state-of-the-art methods}.} The results are shown in Figures \ref{fig:dblp}, \ref{fig:digg}, and \ref{fig:facebook}.

\begin{figure}[!h] 
    \centering
    \includegraphics[width=1\linewidth]{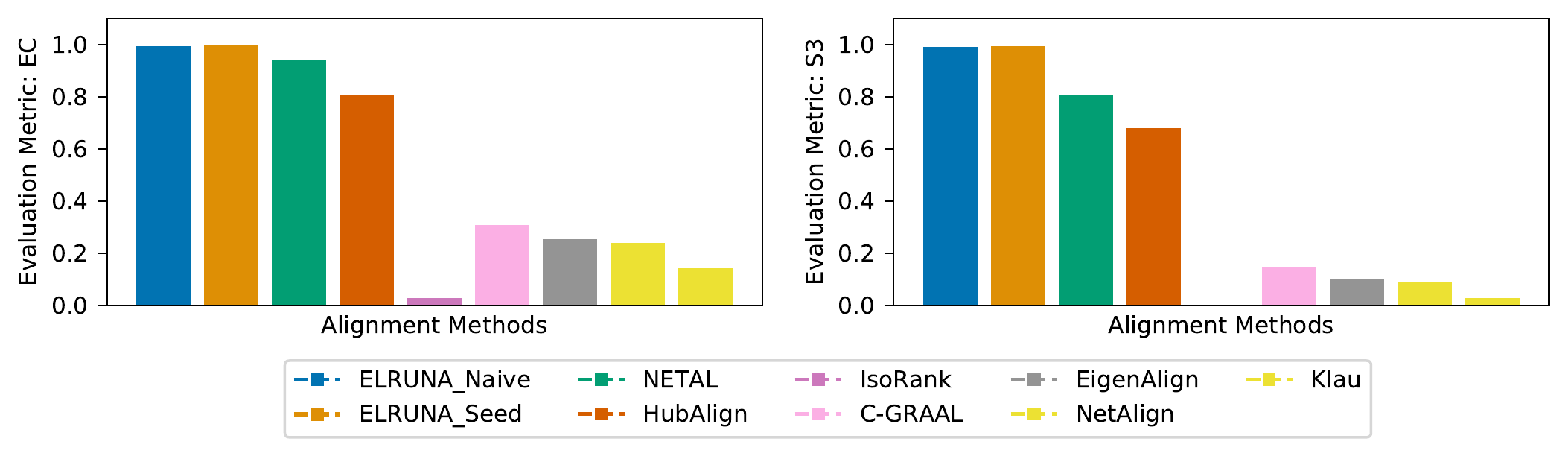}
    \caption{Alignment quality comparison on \texttt{dblp} networks}
    \label{fig:dblp}
\end{figure}

\begin{figure}[!h] 
    \centering
    \includegraphics[width=1\linewidth]{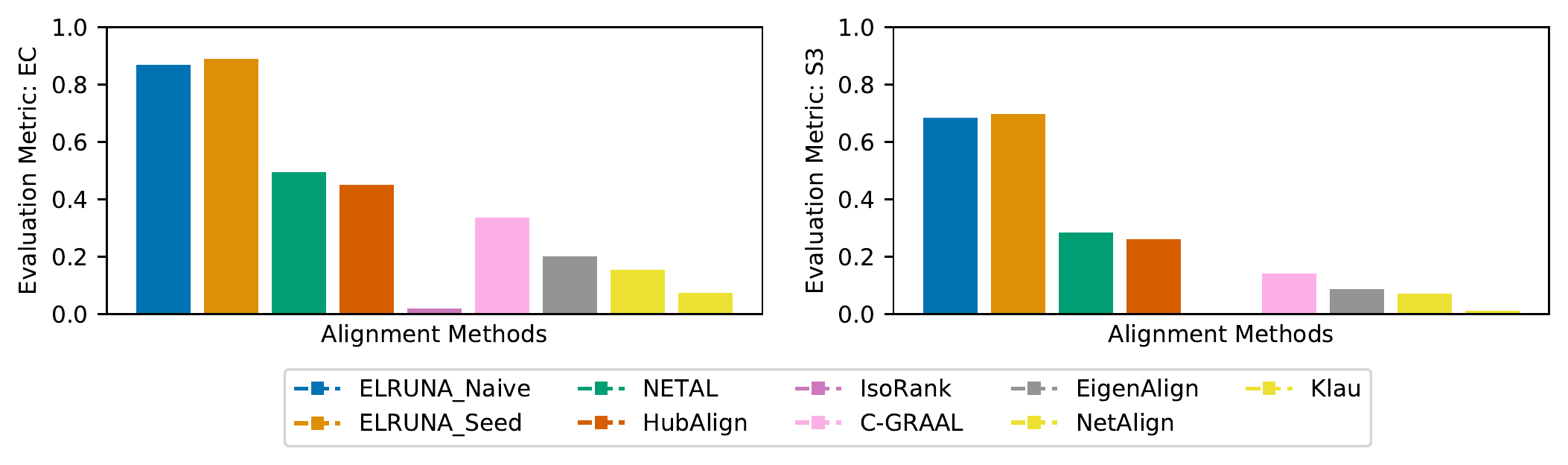}
    \caption{Alignment quality comparison on \texttt{digg} networks}
    \label{fig:digg}
\end{figure}

\begin{figure}[!h] 
    \centering
    \includegraphics[width=1\linewidth]{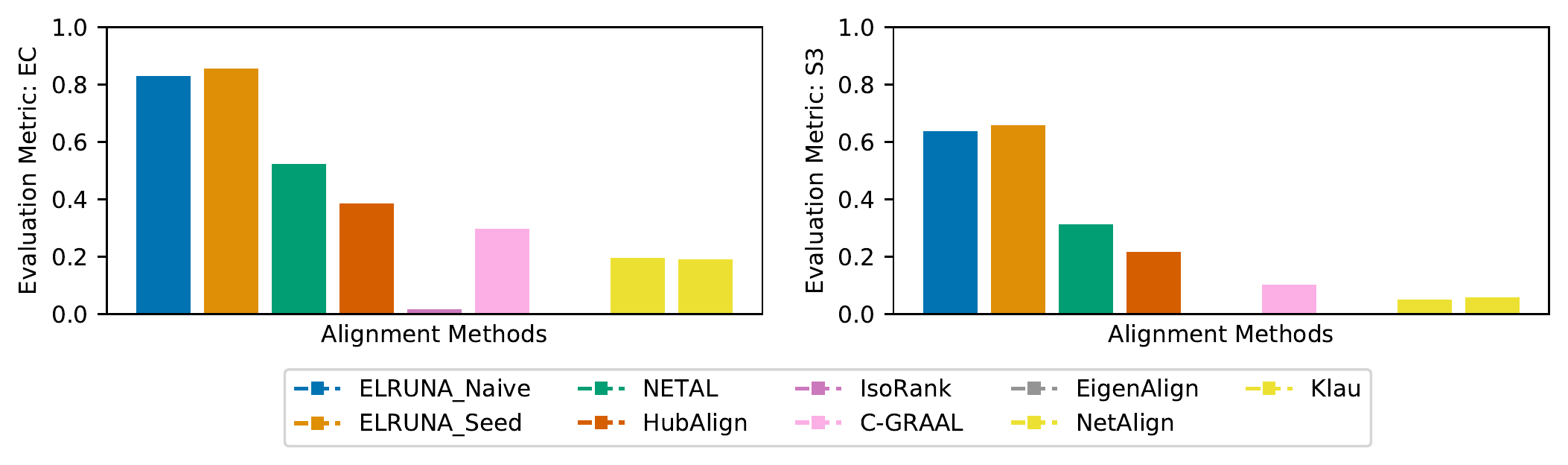}
    \caption{Alignment quality comparison on \texttt{facebook} networks}
    \label{fig:facebook}
\end{figure}

\par \textbf{Results.} \algoname{elruna} outperforms all other baselines on \texttt{DBLP}, \texttt{digg}, and \texttt{facebook} networks. For the \texttt{dblp} networks, even though the $EC$ value of \texttt{NETAL} is close to \algoname{elruna} (differs by $5.7734\%$), the difference between their $S^3$ score is more significant, with $\algoname{elruna}$ surpassing $\texttt{NETAL}$ by $18.7728\%$. For the \texttt{digg} and \texttt{facebook} networks, \texttt{\algoname{elruna}\_Naive} achieves a 175.88\% and a 158.45 \% increase in the $EC$ score over \texttt{NETAL}, respectively. As of $S^3$, \texttt{\algoname{elruna}\_Naive} achieves and a 239.65\% and a 203.8\% increase over \texttt{NETAL}. At the same time, the $EC$ produced by \texttt{\algoname{elruna}\_Naive} is 2 to 11 times higher than for other baselines. \texttt{\algoname{elruna}\_Seed} outperforms \texttt{\algoname{elruna}\_Naive} (and therefore all baselines) with improvement of $EC$ up to 2.76\% and improvement of $S^3$ up to 2.1\% over \texttt{\algoname{elruna}\_Naive}. 

\par This experiment further demonstrates the superiority of \algoname{elruna} in identifying the underlying similar topology of networks and discovering correspondences of nodes.

\subsection{Quality-Speed Trade-off and Scalability}
As shown in the preceding section, \algoname{elruna} significantly outperforms other methods in terms of alignment quality. In this section, we study the quality-speed trade-off of \algoname{elruna} against baselines. Then, we evaluate the scalability of \algoname{elruna}.

\subsubsection{Quality-Speed Trade-off} We first evaluate the trade off by running each algorithm on the \texttt{bio\_1}, \texttt{bio\_2}, \texttt{erdos}, \texttt{retweet}, and \texttt{social} networks under the highest noise levels ($p = 0.25$). That is, we align each $G_1$ with its corresponding $G_2^{(0.25)}$, then record the running time and alignment quality of each algorithm. In addition, we perform the same experiment on two pairs of networks from the homogeneous testing case: \texttt{digg} and \texttt{facebook} networks.

\par We measure the running time (in seconds) and alignment quality ($EC$ and $S^3$) of \texttt{\algoname{elruna}\_Naive} and \texttt{\algoname{elruna}\_Seed} with an incremental number of iterations. That is, we run the algorithm several times until convergence, each time with one additional iteration. For all other methods, we do not perform this incremental-iteration approach because  either they cannot specify the number of iterations or the alignment quality is significantly lower than that of \algoname{elruna}. The results are shown in Figures \ref{fig:time_bio}--\ref{fig:time_facebook}. For clarification, each dot on the \texttt{\algoname{elruna}\_Naive} and \texttt{\algoname{elruna}\_Seed} lines is one measurement after the termination of the algorithm under a particular number of iterations. Two adjacent dots are two measurements that differ by one additional iteration.

\par Note that the running times of \texttt{C-GRALL} and \texttt{EigenAlign} are not included because either they have crashed (as described in the preceding section) or both ran for several hours, which are not comparable to other methods. In addition, the running times of \texttt{REGAL} are not included for the \texttt{digg} and \texttt{facebook} networks because the implementation of \texttt{REGAL} does not support alignment between networks with different sizes. 

\begin{figure}[!h] 
    \centering
    \includegraphics[width=0.8\linewidth]{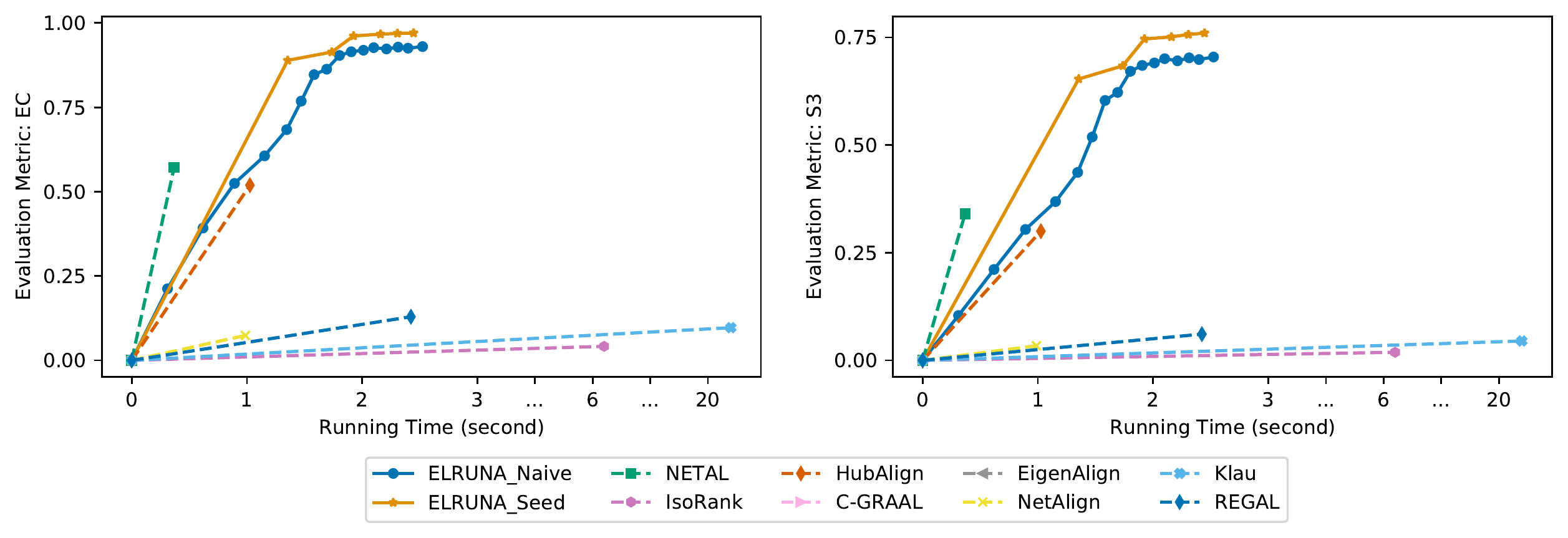}
    \caption{Quality-time comparison on \texttt{bio\_1} networks \runningtime}
    \label{fig:time_bio}
\end{figure}

\clearpage

\begin{figure}[!h] 
    \centering
    \includegraphics[width=0.8\linewidth]{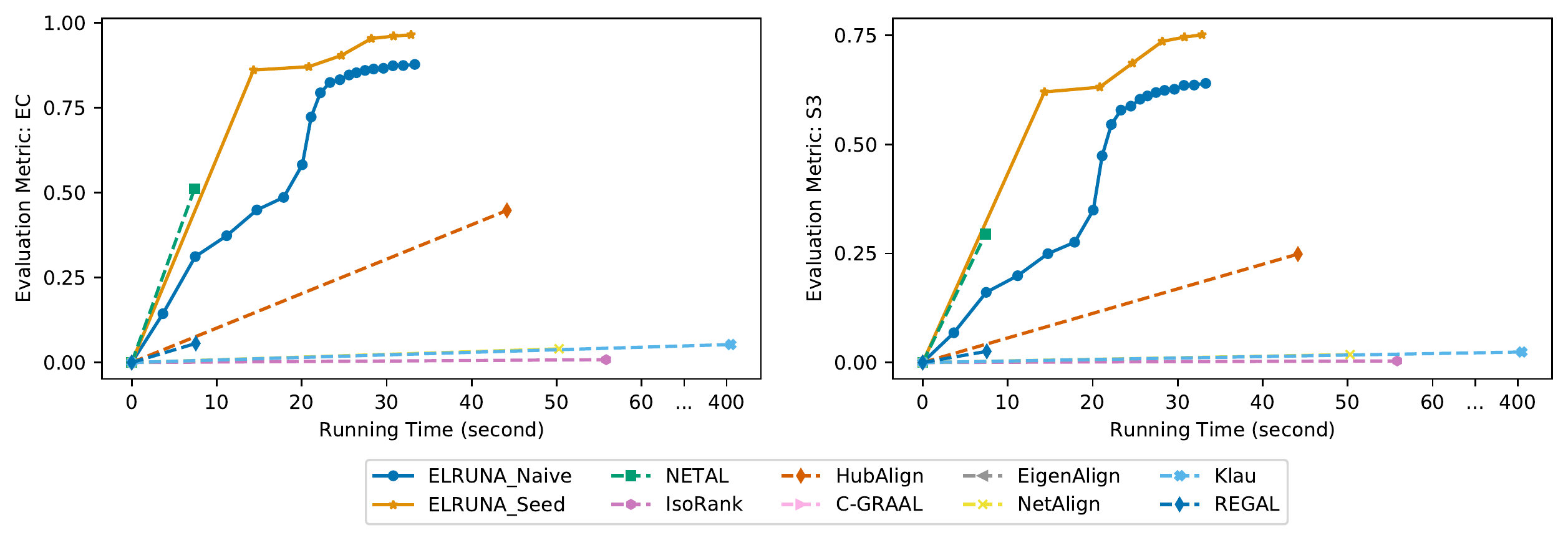}
    \caption{Quality-time comparison on \texttt{bio\_2} networks \runningtime}
    \label{fig:time_bio2}
\end{figure}

\begin{figure}[!h] 
    \centering
    \includegraphics[width=0.85\linewidth]{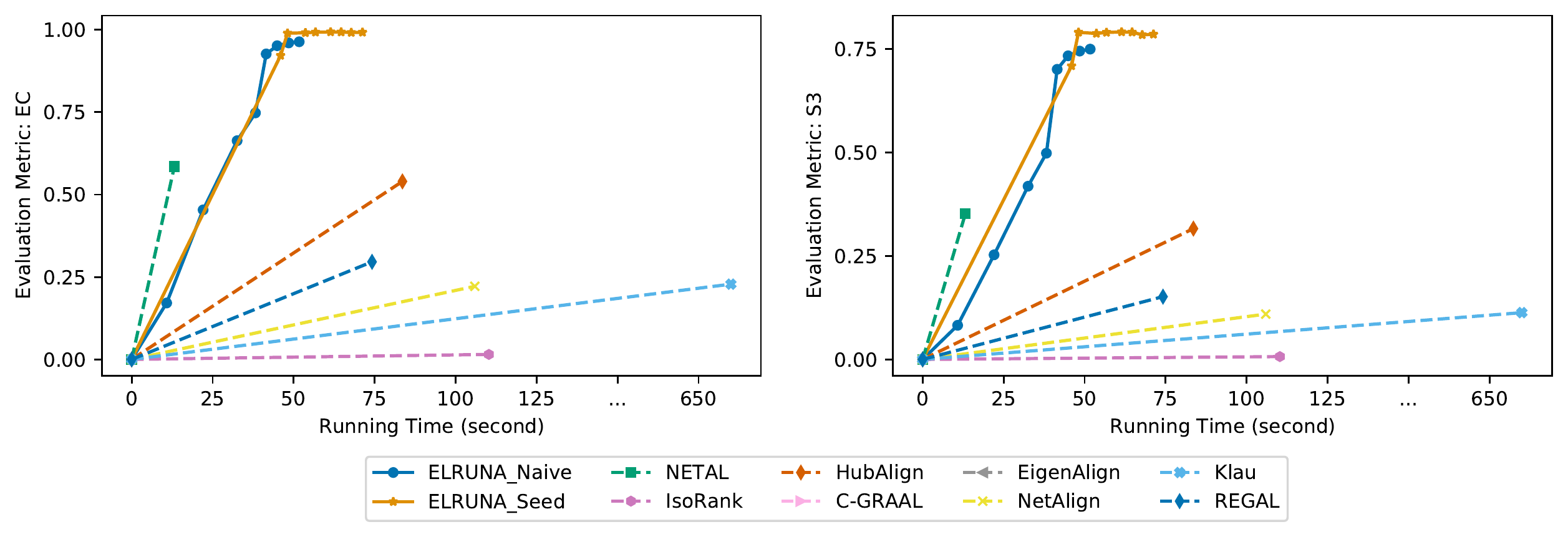}
    \caption{Quality-time comparison on \texttt{erdos} networks \runningtime}
    \label{fig:time_erdos}
\end{figure}

\begin{figure}[!h] 
    \centering
    \includegraphics[width=0.8\linewidth]{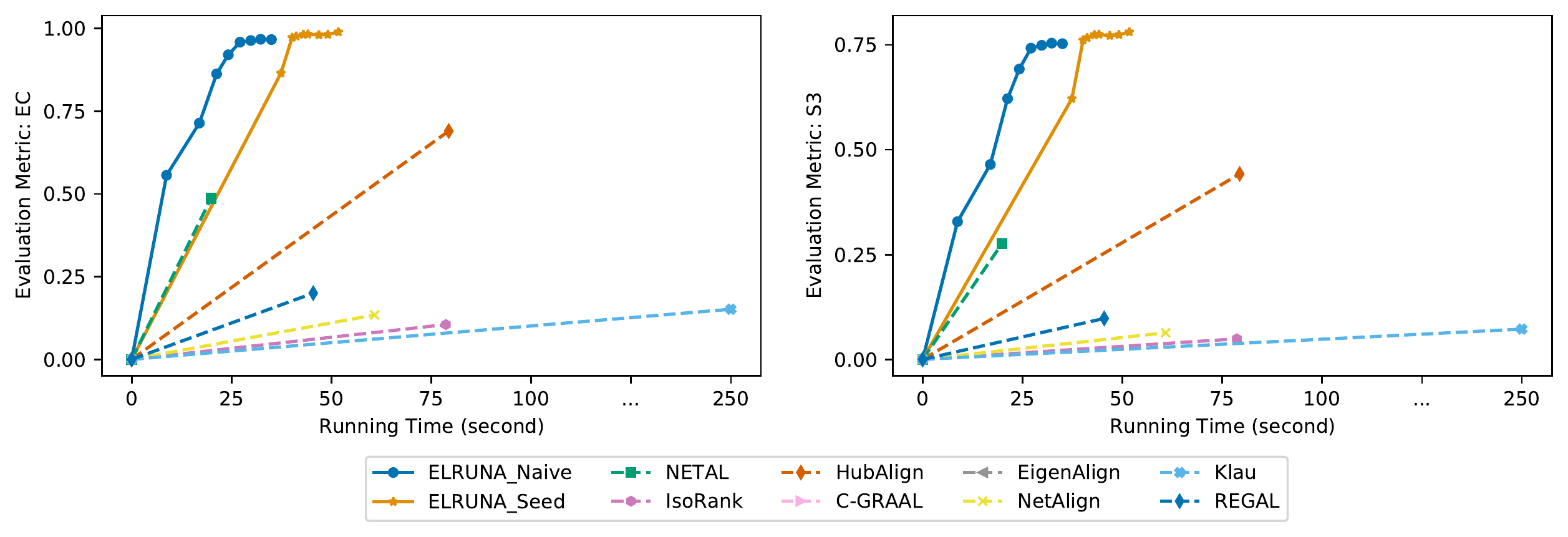}
    \caption{Quality-time comparison on \texttt{retweet\_1} networks \runningtime}
    \label{fig:time_retweet}
\end{figure}

\clearpage

\begin{figure}[!h] 
    \centering
    \includegraphics[width=0.8\linewidth]{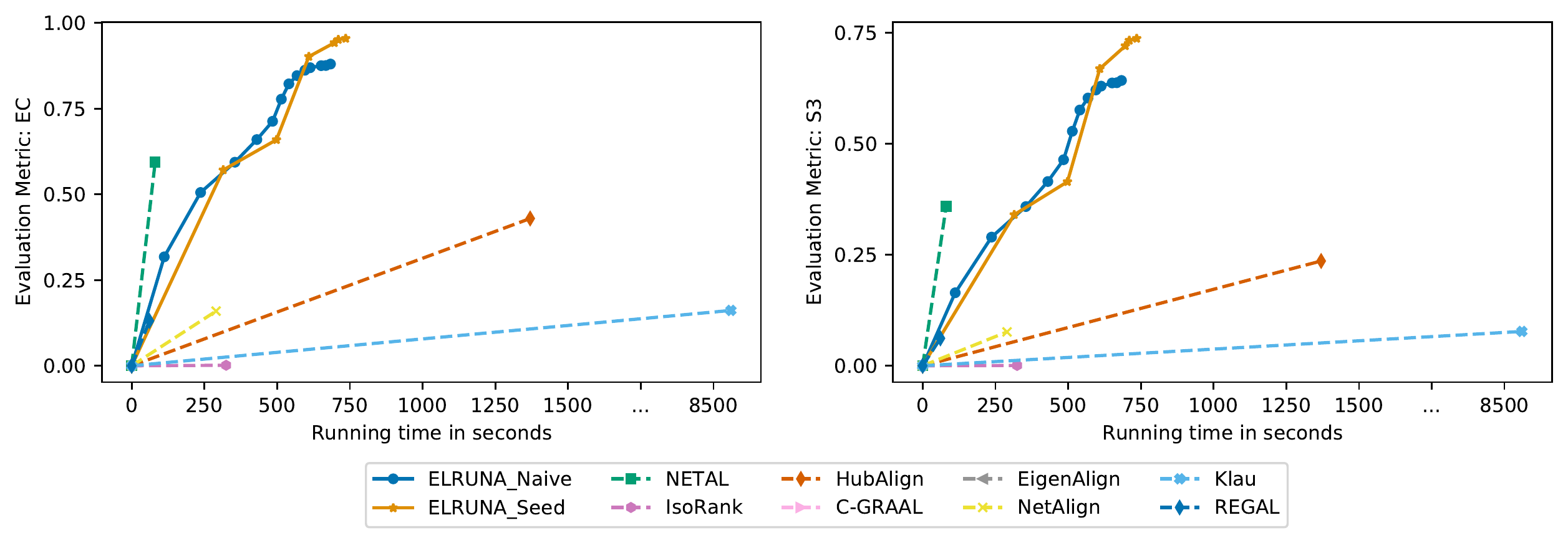}
    \caption{Quality-time comparison on \texttt{social} networks \runningtime}
    \label{fig:time_social}
\end{figure}

\begin{figure}[!h] 
    \centering
    \includegraphics[width=0.8\linewidth]{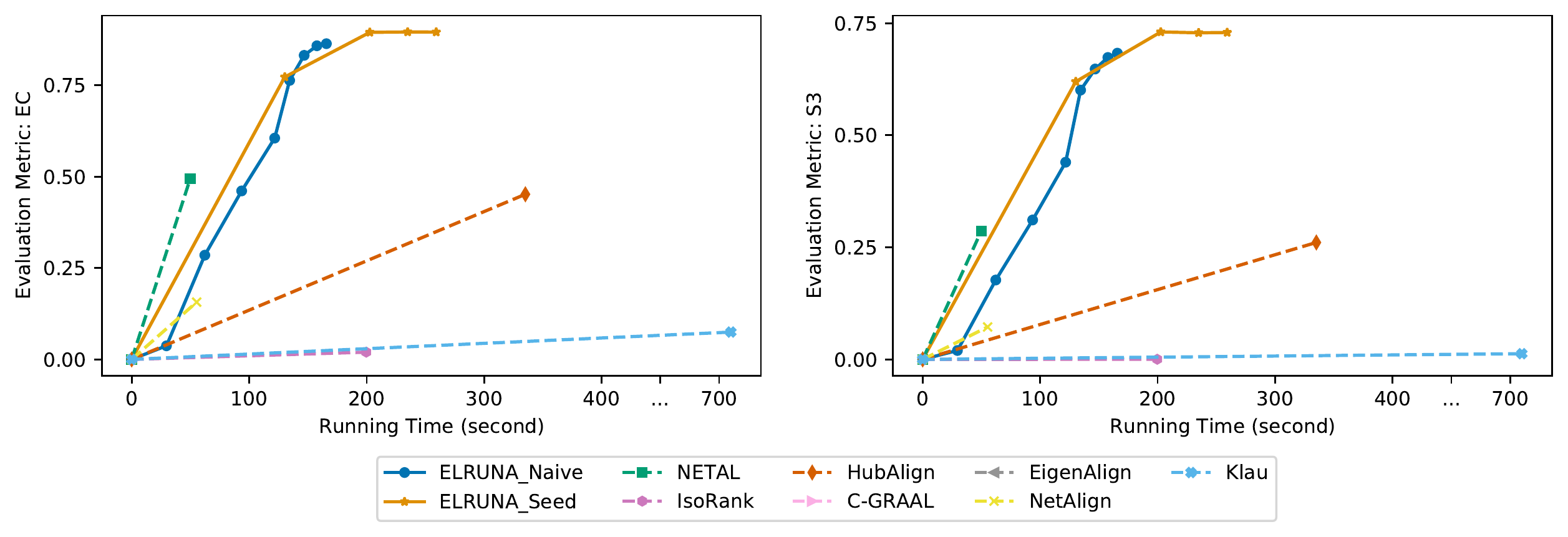}
    \caption{Quality-time comparison on \texttt{digg} networks \runningtime}
    \label{fig:time_digg}
\end{figure}

\begin{figure}[!h] 
    \centering
    \includegraphics[width=0.8\linewidth]{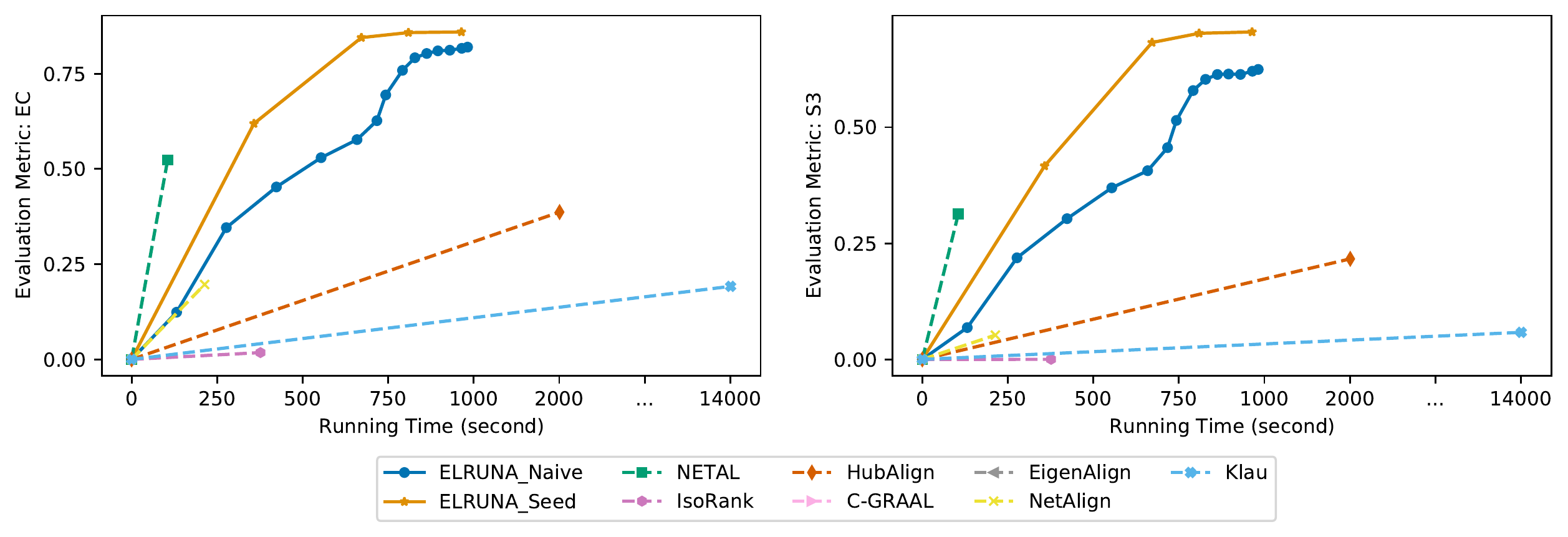}
    \caption{Quality-time comparison on \texttt{facebook} networks \runningtime}
    \label{fig:time_facebook}
\end{figure}

\clearpage

\textbf{Result}. As we have observed, in comparison with \texttt{Klau} and \texttt{HubAlign}, both versions of \algoname{elruna} achieve significantly better alignment results with lower running time. Moreover, \texttt{\algoname{elruna}\_Naive} and \texttt{\algoname{elruna}\_Seed} always have intermediate states (at some $k$th iteration) that have running times  similar to or lower  than those of \texttt{Netalign}, \texttt{REGAL}, and \texttt{IsoRank} but produce much better results. 

\par As for \texttt{NETAL}, we observe that \texttt{\algoname{elruna}\_Naive} always has an intermediate state with similar alignment quality and slightly higher running time. Also, \texttt{\algoname{elruna}\_Seed} always has an intermediate state with better alignment quality and slightly higher running time than those of \texttt{NETAL}. However, both proposed methods can further improve the alignment quality greatly beyond the intermediate state, whereas \texttt{NETAL} and other baselines cannot. In addition, as the two proposed algorithms proceed, each subsequent iteration always takes less time than the previous because the similarities are more defined after each iteration.

\par We observe that \texttt{\algoname{elruna}\_Seed} usually takes fewer iterations and longer running time to converge than does \texttt{\algoname{elruna}\_Naive}. These results are  expected because the \texttt{seed-and-extend} alignment method is more computationally expensive than the \texttt{naive} alignment method. 
\subsubsection{Scalability}
We evaluate the scalability of \texttt{\algoname{elruna}\_Naive} and \texttt{\algoname{elruna}\_Seed}  by running them on networks from the \textit{self-alignment without and under noise} testing case with no noisy edges. That is, we align $G_1$ with $G_2^{(0)}$. The result is shown in Figure \ref{fig:scale}.
\begin{figure}[!h] 
    \centering
    \includegraphics[width=1\linewidth]{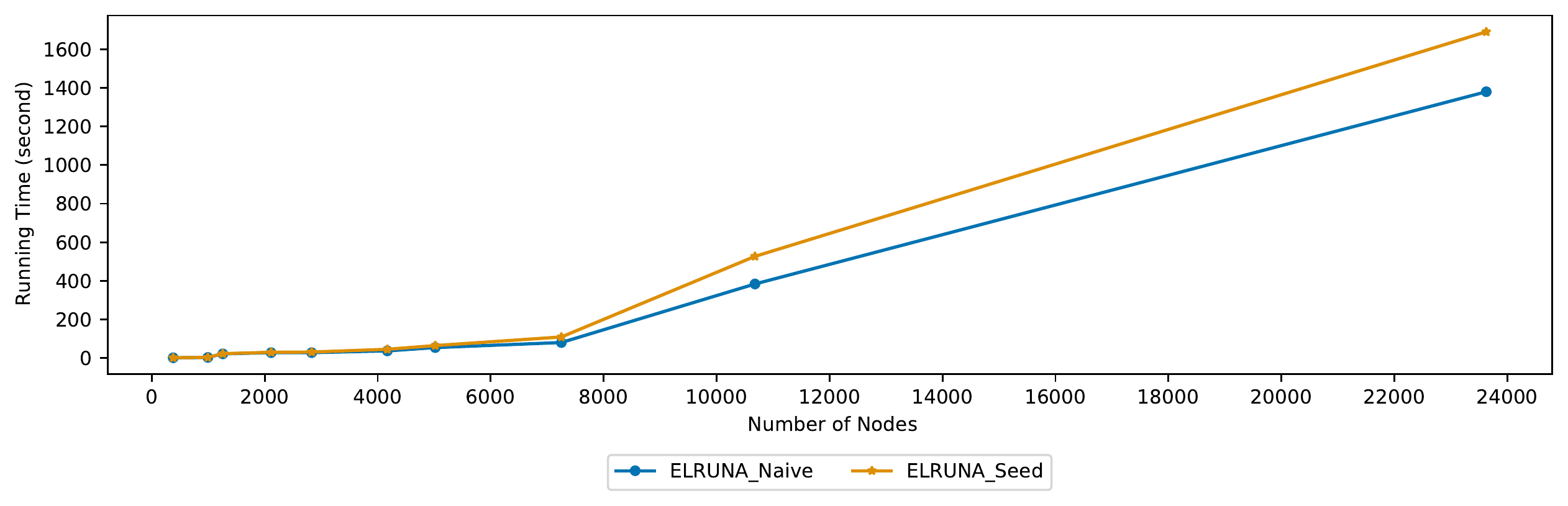}
    \caption{Scalability of \algoname{elruna}}
    \label{fig:scale}
\end{figure}

\par \textbf{Results.} We observe that the running time of both versions of \algoname{elruna} is quadratic with respect to the number of nodes in networks.

\subsection{Limitations: Alignment between Heterogeneous Networks}
In this section, we discuss the limitations of \algoname{elruna}. We compare \algoname{elruna} with baselines on 5 pairs of networks where each pair consists of two networks from different domains. The first 3 pairs are social networks collected in ~\cite{final}, and the other two pairs are biological networks~\cite{saraph2014magna}. The details of these networks are listed in Table \ref{tab:hete}.

\begin{center}
\captionof{table}{Datasets for test case: \textit{Alignment between heterogeneous networks}} \label{tab:hete} 
 \begin{tabular}{||l c c||} 
 \hline
\textbf{Doamin} & $n$ & $m$ \\ [0.5ex] 
 \hline 
 
\texttt{Offline} vs \texttt{Online} & 1,118 vs 3,906 & 1,511 vs 8,164 \\ \hline
 
\texttt{Flickr} vs \texttt{Lastfm} & 12,974 vs 15,436 & 16,149 vs 16,319  \\ \hline
 
\texttt{Flickr} vs \texttt{Myspace} & 6,714 vs 10,693 & 7,333 vs 10,686 \\ \hline

\texttt{Syne} vs \texttt{Yeast} & 1,837 vs 1,994 & 3,062 vs 15,819  \\ \hline

\texttt{Ecoli} vs \texttt{Yeast} & 1,274 vs 1,994 & 3,124 vs 15,819 \\ \hline
 
\end{tabular}
\end{center}

In this experiment, we demonstrate the results without local search; that is,  we show how \algoname{elruna} outperforms the  state-of-the-art methods. Results are shown in Figures \ref{fig:douban}--\ref{fig:yeast}. Because each pair of networks does not have underlying isomorphic subgraphs, \emph{the optimal objective is not known}, and the highest EC is not 1.

\par We note that pairs of networks in this testing case do not have structurally similar underlying subgraphs. In fact, their topology could be very distinct from each other, given their different domains. Also, as we stated at the beginning of the experiment section, this comparison scenario is usually used by \emph{attributed} network alignment algorithms. Therefore, it is not exactly what we solve with our formulation that relies solely on graph structure.

\par \textbf{Results.}  
We observe that \texttt{\algoname{elruna}\_Naive} did not perform well in this testing case  due to the naive alignment method that it uses. As for \texttt{\algoname{elruna}\_Seed}, it achieves a 7.85\% and a 11.41\% increase in $EC$ over  \texttt{HubAlign} and \texttt{NETAL}, respectively, for the pairs \texttt{offline} and \texttt{online} networks. Regarding  other pairs of networks, the \texttt{\algoname{elruna}\_Seed}'s improvements of $EC$ \texttt{HubAlign} and \texttt{NETAL} are statistically insignificant. The results suggest that we need to enhance our algorithm in order to perform better under this testing case. One future direction is to extend \algoname{elruna} to handle attributed networks. 

\par Another limitation is that \algoname{elruna} does not rely on the predefined vertex similarities; therefore, when such information is given, \algoname{elruna} cannot utilize it to achieve better performance. For future work, we want to augment \algoname{elruna}, which considers the prior similarities between vertices.

\begin{figure}[!h] 
    \centering
    \includegraphics[width=0.5\linewidth]{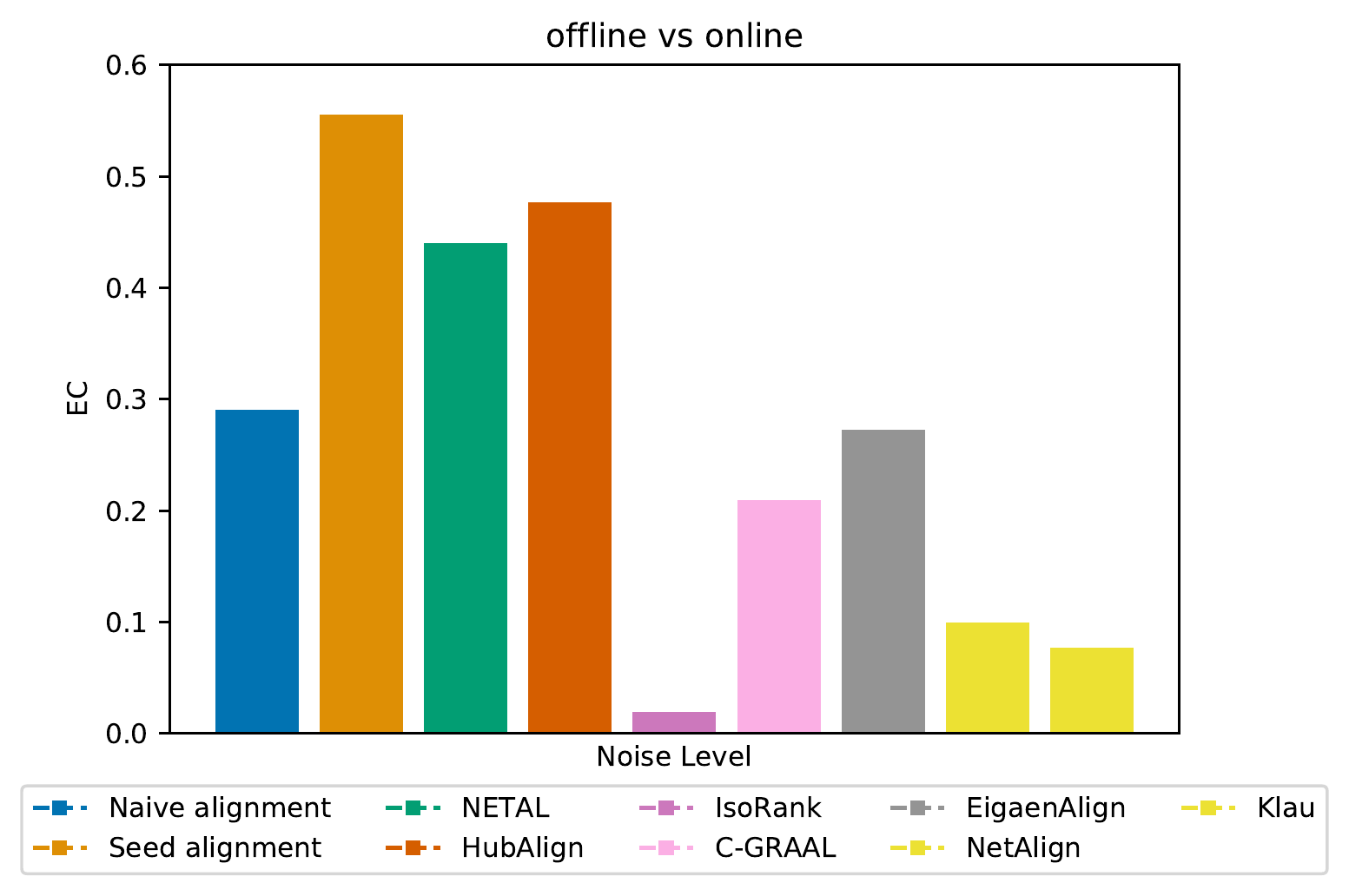}
    \caption{\texttt{offline} vs \texttt{online} networks}
    \label{fig:douban}
\end{figure}

\begin{figure}[!h] 
    \centering
    \includegraphics[width=0.6\linewidth]{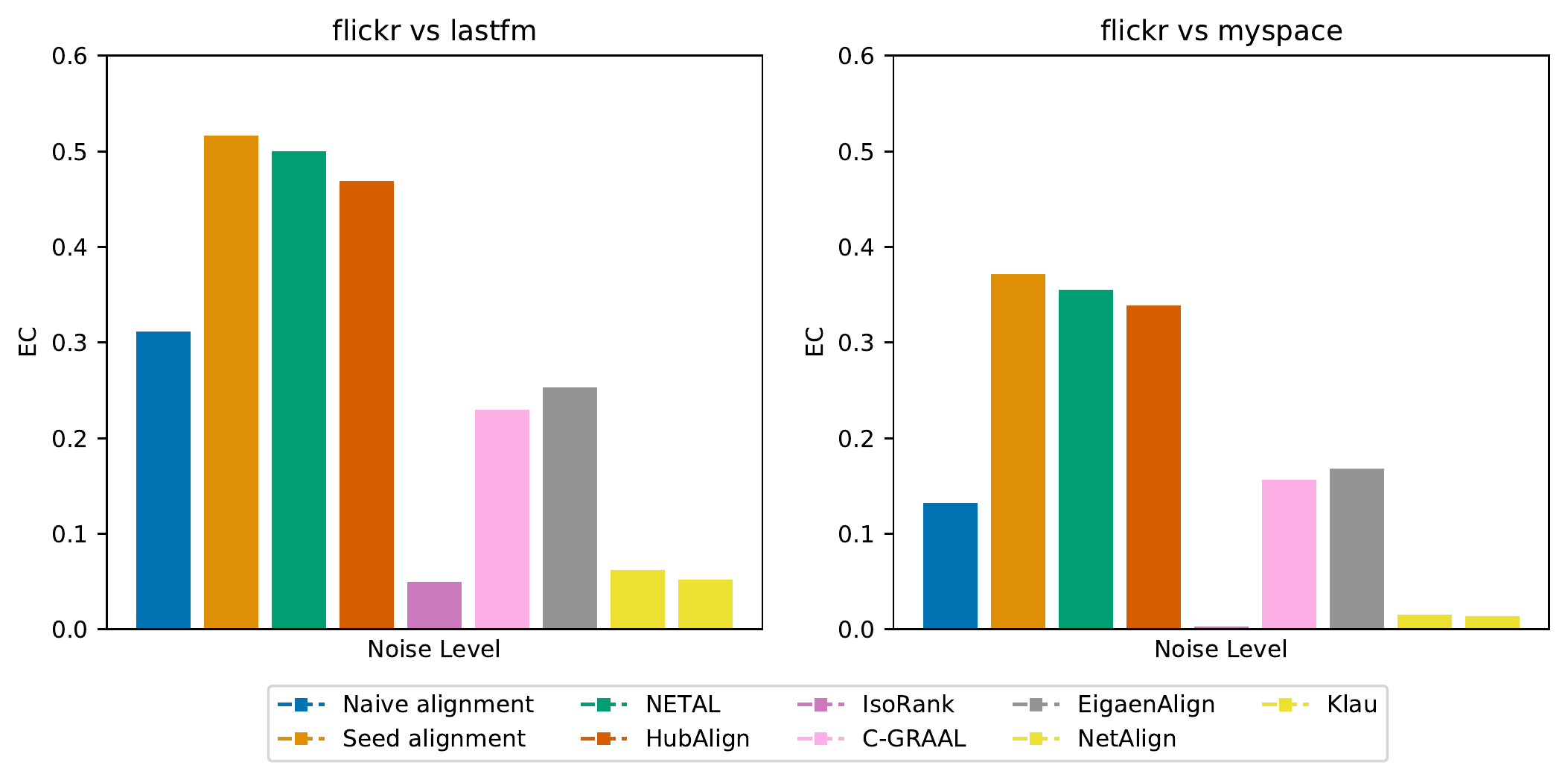}
    \caption{\texttt{Flickr} vs \texttt{Lastfm} and \texttt{Flickr} vs \texttt{Myspace} networks}
    \label{fig:flickr}
\end{figure}
\begin{figure}[!h] 
    \centering
    \includegraphics[width=0.6\linewidth]{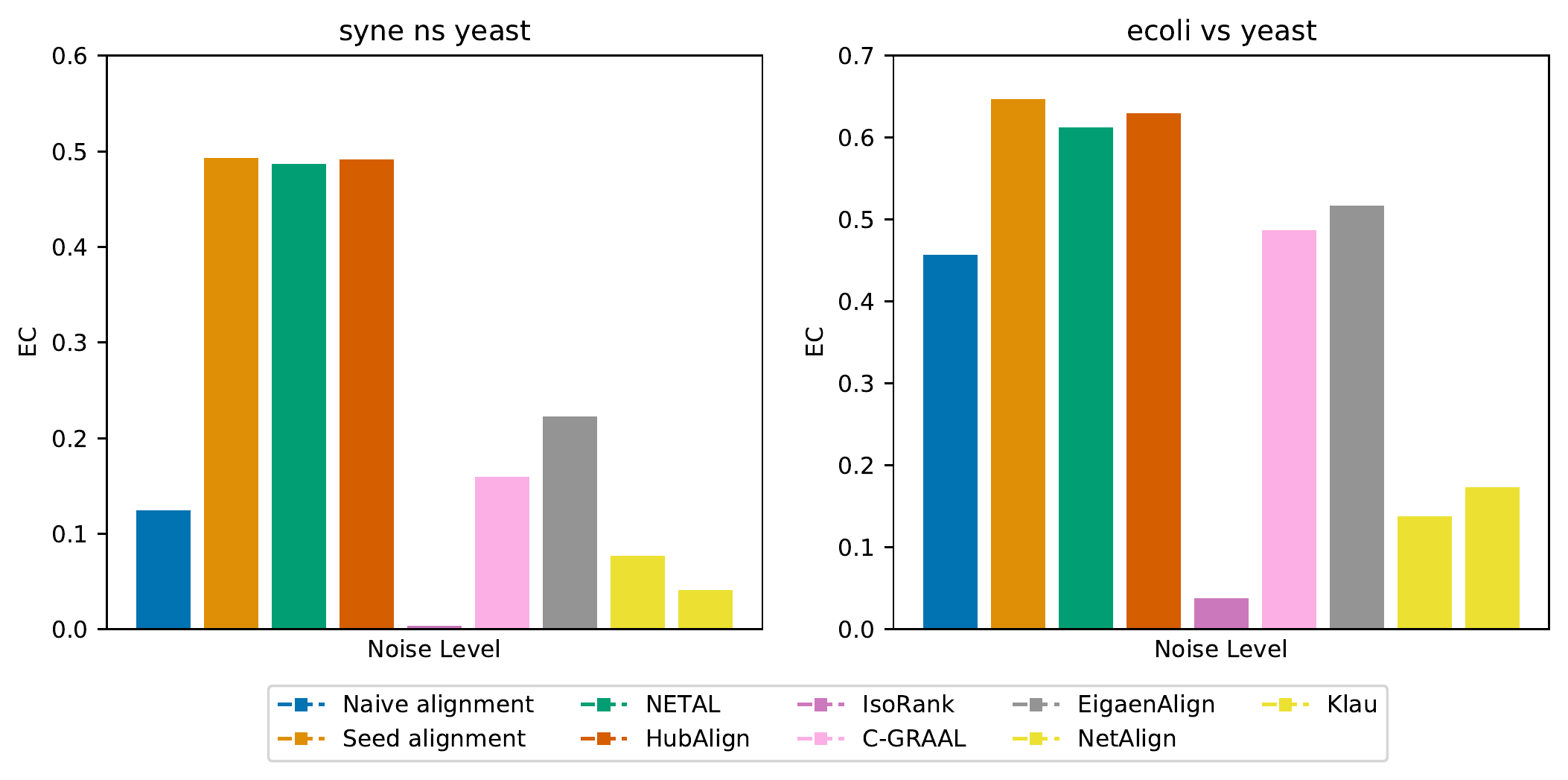}
    \caption{\texttt{Syne} vs \texttt{Yeast} and \texttt{Ecoli} vs \texttt{Yeast} networks}
    \label{fig:yeast}
\end{figure}

\subsection{Evaluation of \algoname{Rawsem}}
In this section, we compare the performance of the proposed \algoname{Rawsem} against the \texttt{baseline local search}. To evaluate their efficiency and effectiveness, we apply both methods as the postprocessing steps for \texttt{\algoname{elruna}\_Naive}. We align real-world networks from the \textit{self-alignment without and under noise} testing case under the highest noise level. That is, $G_1$ is aligned with $G_2^{(0.25)}$. For each method we run it $20$ times on each pair of networks and  measure its average running time, the average number of iterations to reach an optimum, and the average improved alignment quality. The results are summarized in Table \ref{tab:rawsem}.

\textbf{Results.} Clearly, \algoname{Rawsem} outperforms the \texttt{Baseline local search} in terms of  efficiency and effectiveness. In particular, \algoname{Rawsem} achieves an increase of up to 13 times and 11 times for $EC$ and $S^3$, respectively, over the \texttt{Baseline local search}. Additionally, the number of iterations of the \texttt{Baseline local search} is orders of magnitude larger than that of \algoname{Rawsem}.  This experiment shows that \algoname{Rawsem} can boost the alignment quality within seconds, thus making it a great candidate for a postprocessing step of network alignment algorithms.

\par We observe that \algoname{Rawsem} does not always raise the objective to the global optima. This fact suggests that we should combine our selection rule with different local search methods, such as simulated annealing, to further enhance its performance. This is a further direction. 

\begin{center}
\captionof{table}{\algoname{Rawsem} vs \texttt{Baseline} as postprosessing steps} \label{tab:rawsem} 
\begin{tabular}{|l|rr|rr|}
\hline
\multicolumn{1}{|c|}{} & \multicolumn{2}{c|}{\texttt{co-autho\_1}} & \multicolumn{2}{c|}{\texttt{bio\_1}}\\
\cline{2-5}
\multicolumn{1}{|c|}{} & \algoname{Rawsem} & \texttt{Baseline} & \algoname{Rawsem} & \texttt{Baseline}\\

No. of iterations &  1,301  &28,091  & 1,817
 & 32,274\\
Time (in seconds) & 0.093 & 2.052 & 0.095 & 2.327   \\

Improved $EC$ & 2.3\% & 0.982\% &3.231\% &1.073\%\\

Improved $S^3$ &4.91\% & 1.175\% &5.193\% & 1.91\%\\

\hline
\multicolumn{1}{|c|}{} & \multicolumn{2}{c|}{\texttt{econ}} & \multicolumn{2}{c|}{\texttt{router}}\\
\cline{2-5}

\multicolumn{1}{|c|}{} & \algoname{Rawsem} & \texttt{Baseline} & \algoname{Rawsem} & \texttt{Baseline}\\

No. of iterations &  1,580 & 50,000 & 798
 & 42,367\\
Time (in seconds) & 0.11 & 3.324 & 0.077 & 3.049 \\

Improved $EC$ & 0.506\% & 0\% &1.406\% &0.441\%\\

Improved $S^3$ &0.724\% & 0\% &2.991\% & 0.892\%\\
\hline

\multicolumn{1}{|c|}{} & \multicolumn{2}{c|}{\texttt{bio\_2}} & \multicolumn{2}{c|}{\texttt{retweet\_1}}\\
\cline{2-5}

\multicolumn{1}{|c|}{} & \algoname{Rawsem} & \texttt{Baseline} & \algoname{Rawsem} & \texttt{Baseline}\\

No. of iterations &  3,069 & 61,290 & 4,415 & 60,272\\

Time (in seconds) & 0.191 & 4.48 & 0.217 & 4.411 \\

Improved $EC$ & 5.477\% & 0.61\% &1.241\% & 0.392\%\\

Improved $S^3$ &7.017\% & 1.326\% &2.019\% & 0.673\%\\

\hline
\multicolumn{1}{|c|}{} & \multicolumn{2}{c|}{\texttt{erods}} & \multicolumn{2}{c|}{\texttt{retweet\_2}}\\
\cline{2-5}

\multicolumn{1}{|c|}{} & \algoname{Rawsem} & \texttt{Baseline} & \algoname{Rawsem} & \texttt{Baseline}\\

No. of iterations & 4,701 & 88,221 & 3,320 & 72,392\\

Time (in seconds) & 0.323 & 6.25 & 0.204 & 5.19 \\

Improved $EC$ & 2.91\% & 0.31\% &5.72\% & 0.437\%\\

Improved $S^3$ &4.801\% & 0.781\% &10.08\% & 0.901\%\\
\hline

\multicolumn{1}{|c|}{} & \multicolumn{2}{c|}{\texttt{social}} & \multicolumn{2}{c|}{\texttt{google+}}\\
\cline{2-5}

\multicolumn{1}{|c|}{} & \algoname{Rawsem} & \texttt{Baseline} & \algoname{Rawsem} & \texttt{Baseline}\\

No. of iterations & 7,701 & 100,297 & 11,928 & 152,116\\

Time (in seconds) & 1.953 & 7.855 & 3.04 & 10.238 \\

Improved $EC$ & 4.903\% & 0.631\% &3.29\% & 0.723\%\\

Improved $S^3$ &9.29\% & 1.01\% &7.81\% & 1.59\%\\
\hline
\end{tabular}
\end{center}

% -----------------------------------
% -           Summary               -
% -----------------------------------
\section{Conclusion and Future Work}
In this paper, we propose \algoname{elruna}, an iterative network alignment algorithm based on  elimination rules. We also introduce \algoname{rawsem}, a novel random-walk-based selection rule for local search schemes that decreases the number of iterations needed to reach a local or global optimum. We conducted extensive experiments and demonstrate the superiority of \algoname{elruna} and \algoname{rawsem}. 

For future work, we aim to (1) improve the performance of \algoname{elruna} on aligning regular graphs, (2) extend \algoname{elruna} on aligning dense networks, and (3) develop advanced local search schemes to further reduce the number of iterations. Another possible extension of this work is to adopt the proposed algorithm to run on quantum computers as we have done  in our previous work~\cite{shaydulin2018community, ushijima2019multilevel, shaydulin2019network, shaydulin2019quantum}.

\section*{Acknowledgement}
We  thank Gunnar W. Klau for helping us with the source code of \texttt{Klau}'s algorithm.
Also, we thank Clemson University for its allotment of compute time on the Palmetto cluster. Argonne National Laboratory's work was supported by the U.S. Department of Energy, Office of Science, under contract DE-AC02-06CH11357. We would like to thank three anonymous reviewers whose insightful comments helped to considerably improve this paper.

\clearpage

\bibliographystyle{plain}
\bibliography{main}

\begin{thebibliography}{10}

\bibitem{barabasi}
A.~Barab{\'a}si and R.~Albert.
\newblock Emergence of scaling in random networks.
\newblock {\em science}, 286(5439):509--512, 1999.

\bibitem{netalign}
M.~Bayati, M.~Gerritsen, D.~Gleich, A.~Saberi, and Y.~Wang.
\newblock Algorithms for large, sparse network alignment problems.
\newblock In {\em 2009 Ninth IEEE International Conference on Data Mining},
  pages 705--710, 2009.

\bibitem{boguna2004models}
Mari{\'a}n Bogun{\'a}, Romualdo Pastor-Satorras, Albert D{\'\i}az-Guilera, and
  Alex Arenas.
\newblock Models of social networks based on social distance attachment.
\newblock {\em Physical review E}, 70(5):056122, 2004.

\bibitem{digg}
Munmun De~Choudhury, Hari Sundaram, Ajita John, and Dor{\'e}e~Duncan Seligmann.
\newblock Social synchrony: Predicting mimicry of user actions in online social
  media.
\newblock In {\em 2009 International conference on computational science and
  engineering}, volume~4, pages 151--158. IEEE, 2009.

\bibitem{du2017first}
Boxin Du, Si~Zhang, Nan Cao, and Hanghang Tong.
\newblock First: Fast interactive attributed subgraph matching.
\newblock In {\em Proceedings of the 23rd ACM SIGKDD International Conference
  on Knowledge Discovery and Data Mining}, pages 1447--1456, 2017.

\bibitem{feizi2019spectral}
Soheil Feizi, Gerald Quon, Mariana Mendoza, Muriel Medard, Manolis Kellis, and
  Ali Jadbabaie.
\newblock Spectral alignment of graphs.
\newblock {\em IEEE Transactions on Network Science and Engineering}, 2019.

\bibitem{survey}
P.~Hiram Guzzi and T.~Milenkovi{\'c}.
\newblock Survey of local and global biological network alignment: the need to
  reconcile the two sides of the same coin.
\newblock {\em Briefings in bioinformatics}, 19(3):472--481, 2017.

\bibitem{modulealign}
S.~Hashemifar, J.~Ma, H.~Naveed, S.~Canzar, and J.~Xu.
\newblock {ModuleAlign}: module-based global alignment of protein--protein
  interaction networks.
\newblock {\em Bioinformatics}, 32(17):i658--i664, 2016.

\bibitem{hubalign}
S.~Hashemifar and J.~Xu.
\newblock Hubalign: an accurate and efficient method for global alignment of
  protein--protein interaction networks.
\newblock {\em Bioinformatics}, 30(17):i438--i444, 2014.

\bibitem{heimann2018hashalign}
Mark Heimann, Wei Lee, Shengjie Pan, Kuan-Yu Chen, and Danai Koutra.
\newblock Hashalign: Hash-based alignment of multiple graphs.
\newblock In {\em Pacific-Asia Conference on Knowledge Discovery and Data
  Mining}, pages 726--739. Springer, 2018.

\bibitem{heimann2018regal}
Mark Heimann, Haoming Shen, Tara Safavi, and Danai Koutra.
\newblock Regal: Representation learning-based graph alignment.
\newblock In {\em Proceedings of the 27th ACM International Conference on
  Information and Knowledge Management}, pages 117--126, 2018.

\bibitem{holme}
P.~Holme and B.~Kim.
\newblock Growing scale-free networks with tunable clustering.
\newblock {\em Physical review E}, 65(2):026107, 2002.

\bibitem{natali}
G~W. Klau.
\newblock A new graph-based method for pairwise global network alignment.
\newblock {\em BMC bioinformatics}, 10(1):S59, 2009.

\bibitem{bigalign}
D.~Koutra, H.~Tong, and D.~Lubensky.
\newblock Big-align: Fast bipartite graph alignment.
\newblock In {\em 2013 IEEE 13th ICDM}, pages 389--398, 2013.

\bibitem{koutra2017individual}
Danai Koutra and Christos Faloutsos.
\newblock Individual and collective graph mining: principles, algorithms, and
  applications.
\newblock {\em Synthesis Lectures on Data Mining and Knowledge Discovery},
  9(2):1--206, 2017.

\bibitem{leskovec2012learning}
Jure Leskovec and Julian~J Mcauley.
\newblock Learning to discover social circles in ego networks.
\newblock In {\em Advances in neural information processing systems}, pages
  539--547, 2012.

\bibitem{liao2009isorankn}
Chung-Shou Liao, Kanghao Lu, Michael Baym, Rohit Singh, and Bonnie Berger.
\newblock {IsoRankN}: spectral methods for global alignment of multiple protein
  networks.
\newblock {\em Bioinformatics}, 25(12):i253--i258, 2009.

\bibitem{liu2017novel}
Yangwei Liu, Hu~Ding, Danyang Chen, and Jinhui Xu.
\newblock Novel geometric approach for global alignment of {PPI} networks.
\newblock In {\em Thirty-First AAAI Conference on Artificial Intelligence},
  2017.

\bibitem{memivsevic2012c}
Vesna Memi{\v{s}}evi{\'c} and Nata{\v{s}}a Pr{\v{z}}ulj.
\newblock C-graal: Common-neighbors-based global graph alignment of biological
  networks.
\newblock {\em Integrative Biology}, 4(7):734--743, 2012.

\bibitem{netal}
B.~Neyshabur, A.~Khadem, S.~Hashemifar, and S.~Arab.
\newblock {NETAL}: a new graph-based method for global alignment of
  protein--protein interaction networks.
\newblock {\em Bioinformatics}, 29(13):1654--1662, 2013.

\bibitem{pagerank}
Lawrence Page, Sergey Brin, Rajeev Motwani, and Terry Winograd.
\newblock The pagerank citation ranking: Bringing order to the web.
\newblock Technical report, Stanford InfoLab, 1999.

\bibitem{ghost}
Rob Patro and Carl Kingsford.
\newblock Global network alignment using multiscale spectral signatures.
\newblock {\em Bioinformatics}, 28(23):3105--3114, 2012.

\bibitem{ron2011relaxation}
Dorit Ron, Ilya Safro, and Achi Brandt.
\newblock Relaxation-based coarsening and multiscale graph organization.
\newblock {\em Multiscale Modeling \& Simulation}, 9(1):407--423, 2011.

\bibitem{nr-sigkdd16}
Ryan~A. Rossi and Nesreen~K. Ahmed.
\newblock An interactive data repository with visual analytics.
\newblock {\em SIGKDD Explor.}, 17(2):37--41, 2016.

\bibitem{safro2006graph}
Ilya Safro, Dorit Ron, and Achi Brandt.
\newblock Graph minimum linear arrangement by multilevel weighted edge
  contractions.
\newblock {\em Journal of Algorithms}, 60(1):24--41, 2006.

\bibitem{safro2015advanced}
Ilya Safro, Peter Sanders, and Christian Schulz.
\newblock Advanced coarsening schemes for graph partitioning.
\newblock {\em Journal of Experimental Algorithmics (JEA)}, 19:1--24, 2015.

\bibitem{saraph2014magna}
Vikram Saraph and Tijana Milenkovi{\'c}.
\newblock {MAGNA}: maximizing accuracy in global network alignment.
\newblock {\em Bioinformatics}, 30(20):2931--2940, 2014.

\bibitem{shaydulin2018community}
R.~Shaydulin, H.~Ushijima-Mwesigwa, I.~Safro, S.~Mniszewski, and Y.~Alexeev.
\newblock Community detection across emerging quantum architectures.
\newblock {\em arXiv preprint arXiv:1810.07765}, 2018.

\bibitem{shaydulin2019quantum}
Ruslan Shaydulin, Hayato Ushijima-Mwesigwa, Christian~FA Negre, Ilya Safro,
  Susan~M Mniszewski, and \textbf{Yuri Alexeev}.
\newblock A hybrid approach for solving optimization problems on small quantum
  computers.
\newblock {\em Computer}, 52(6):18--26, 2019.

\bibitem{shaydulin2019network}
Ruslan Shaydulin, Hayato Ushijima-Mwesigwa, Ilya Safro, Susan Mniszewski, and
  Yuri Alexeev.
\newblock Network community detection on small quantum computers.
\newblock {\em Advanced Quantum Technologies}, page 1900029, 2019.

\bibitem{isorank}
R.~Singh, J.~Xu, and B.~Berger.
\newblock Global alignment of multiple protein interaction networks with
  application to functional orthology detection.
\newblock {\em Proceedings of the National Academy of Sciences},
  105(35):12763--12768, 2008.

\bibitem{ushijima2019multilevel}
Hayato Ushijima-Mwesigwa, Ruslan Shaydulin, Christian~FA Negre, Susan~M
  Mniszewski, Yuri Alexeev, and Ilya Safro.
\newblock Multilevel combinatorial optimization across quantum architectures.
\newblock {\em arXiv preprint arXiv:1910.09985}, 2019.

\bibitem{vijayan2017alignment}
Vipin Vijayan, Dominic Critchlow, and Tijana Milenkovi{\'c}.
\newblock Alignment of dynamic networks.
\newblock {\em Bioinformatics}, 33(14):i180--i189, 2017.

\bibitem{vijayan2015magna++}
Vipin Vijayan, Vikram Saraph, and T~Milenkovi{\'c}.
\newblock {MAGNA++}: Maximizing accuracy in global network alignment via both
  node and edge conservation.
\newblock {\em Bioinformatics}, 31(14):2409--2411, 2015.

\bibitem{facebook}
Bimal Viswanath, Alan Mislove, Meeyoung Cha, and Krishna~P Gummadi.
\newblock On the evolution of user interaction in {Facebook}.
\newblock In {\em Proceedings of the 2nd ACM workshop on Online social
  networks}, pages 37--42, 2009.

\bibitem{approx_qap}
J.~Vogelstein, J.~Conroy, V.~Lyzinski, L.~Podrazik, S.~Kratzer, E.~Harley,
  D.~Fishkind, R.~Vogelstein, and C.~Priebe.
\newblock Fast approximate quadratic programming for graph matching.
\newblock {\em PLOS one}, 10(4):e0121002, 2015.

\bibitem{yasar2018iterative}
Abdurrahman Yasar and {\"U}mit~V {\c{C}}ataly{\"u}rek.
\newblock An iterative global structure-assisted labeled network aligner.
\newblock In {\em Proceedings of the 24th ACM SIGKDD International Conference
  on Knowledge Discovery \& Data Mining}, pages 2614--2623, 2018.

\bibitem{final}
S.~Zhang and H.~Tong.
\newblock Final: Fast attributed network alignment.
\newblock In {\em Proceedings of the 22nd ACM SIGKDD}, pages 1345--1354. ACM,
  2016.

\bibitem{multilevel}
S.~Zhang, H.~Tong, R.~Maciejewski, and T.~Eliassi-Rad.
\newblock Multilevel network alignment.
\newblock In {\em WWW}, pages 2344--2354. ACM, 2019.

\end{thebibliography}

\appendix
\newpage 
\section*{Appendix}
\subsection*{\textbf{Pictorial examples of Rules 1 and 3}}
\begin{figure}[!h] 
    \centering
    \includegraphics[width=0.55\linewidth]{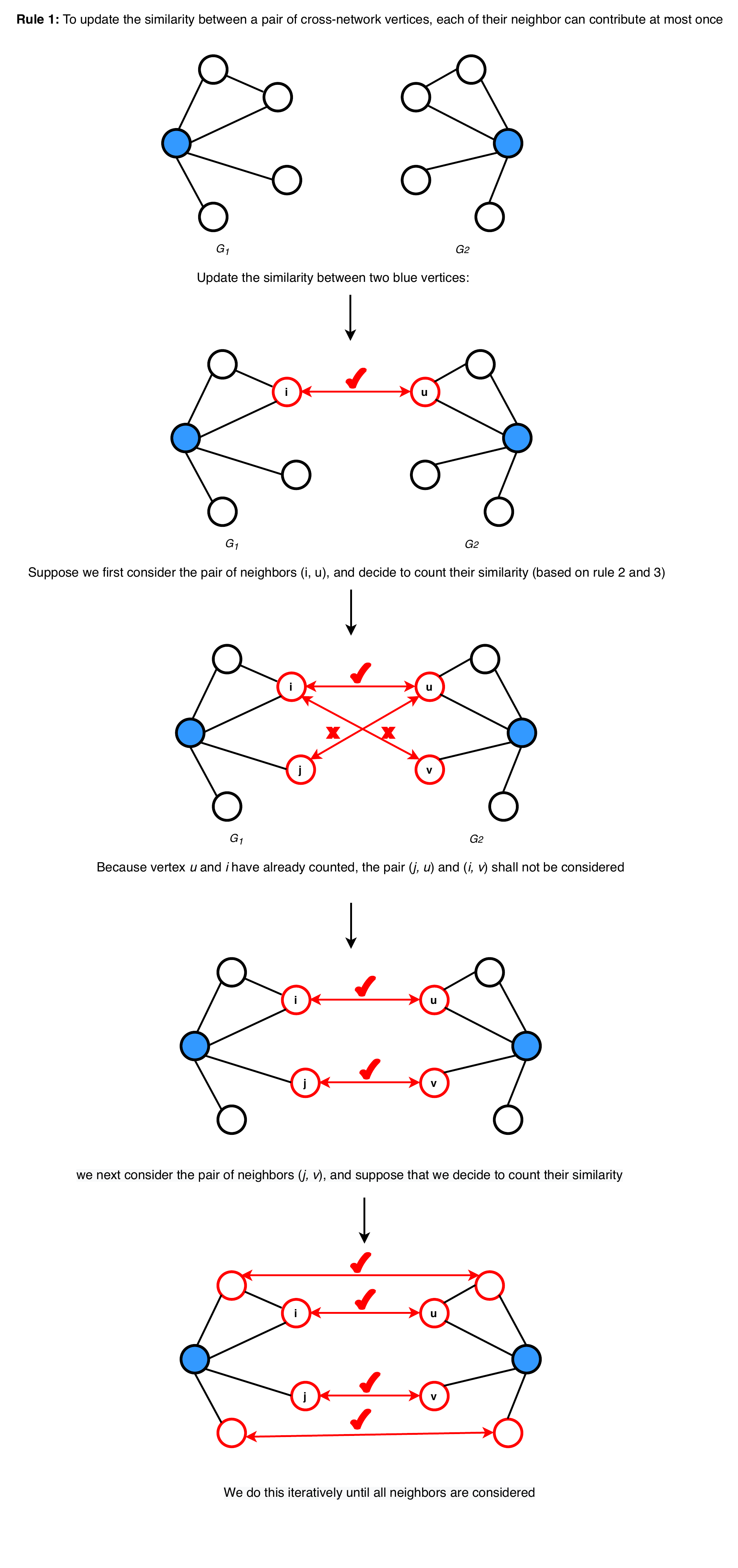}
    \caption{Rule 1 on an example graph}
\end{figure}

\begin{figure}[!h] 
    \centering
    \includegraphics[width=1\linewidth]{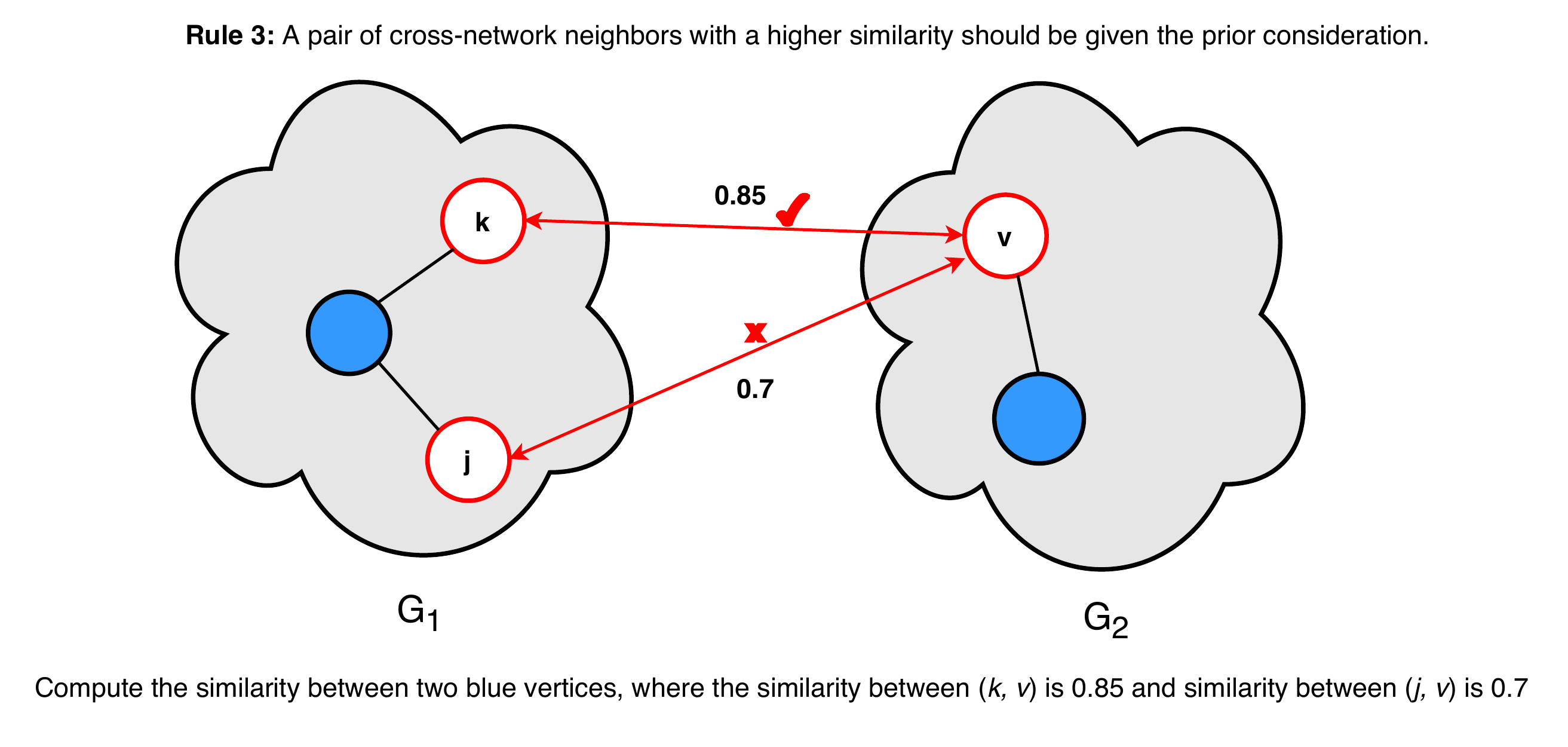}
    \caption{Rule 3 on an example graph}
\end{figure}

\end{document}